\documentclass[a4paper,11pt]{scrartcl}
\usepackage{authblk}

\usepackage{enumerate, comment}
\usepackage[]{graphicx}
\usepackage{bbm}
\usepackage{subfigure, float} 
\usepackage{epstopdf}
\usepackage{lmodern}
\usepackage{caption}

\usepackage{tikz-network}
\usepackage{tikz-cd}

\tikzstyle{mybox} = [draw=black, very thick, diamond, inner ysep=5pt, inner xsep=5pt]
\tikzstyle{myboxS} = [draw=black, thick, rectangle, rounded corners, inner ysep=5pt, inner xsep=5pt]

\tikzstyle{arrow}=[draw, -latex]
\usetikzlibrary{decorations.pathmorphing}

\usepackage{framed, xcolor}
\definecolor{shadecolor}{gray}{0.9}

\tikzstyle{arrow}=[draw, -latex]
\usetikzlibrary{decorations.pathmorphing}
\definecolor{BrickRed}{rgb}{0.8, 0.25, 0.33}
\definecolor{OliveGreen}{rgb}{0.33, 0.42, 0.18}
\definecolor{White}{rgb}{1, 1, 1}
\usepackage{amsmath,amssymb,amsthm}
\newtheorem{lemma}{Lemma}[section]
\newtheorem{proposition}[lemma]{Proposition}

\newtheorem{corollary}[lemma]{Corollary}
\newtheorem{definition}[lemma]{Definition}
\newtheorem{example}[lemma]{Example}
\newtheorem{remark}[lemma]{Remark}

\begin{document}

\title{ Decision-Making Frameworks for Network Resilience \\ --\\  Managing and Mitigating  Systemic (Cyber) Risk}

	\author[a]{Gregor Svindland}
	\author[a]{Alexander Vo\ss}
	\affil[a]{\normalsize Institute of Actuarial and Financial Mathematics \& House of Insurance, Leibniz University Hannover, Germany}

	\date{\today}
	
\maketitle

\begin{abstract}
 {We introduce a decision-making framework tailored for the management of systemic risk in networks. This framework is constructed upon three fundamental components: (1) a set of acceptable network configurations, (2) a set of interventions aimed at risk mitigation, and (3) a cost function quantifying the expenses associated with these interventions.
While our discussion primarily revolves around the management of systemic cyber risks in digital networks, we concurrently draw parallels to risk management of other complex systems where analogous approaches may be adequate.}\end{abstract}\vspace{0.2cm}
\\
\textbf{MSC2020 Subject Classification:} 05C82, 05C90, 90B18, 91B05, 91G45, 93B70

\section{Introduction}\label{sec:intro}

 Modern societies and economies are increasingly dependent on systems that are characterized by complex forms of interconnectedness and interactions between different actors and system entities. Examples of such critical infrastructures are energy grids, transportation and communication systems, financial markets,  and digital systems such as the internet.

However, the complexity of interaction channels also constitutes a major source of risk. 
 The term \textit{systemic risk} refers to the risk arising from the internal characteristics of a system. 
Typically, systemic risk is triggered by some initial failure or disruption in some part of the system, and then propagates and amplifies via various types of channels connecting the system entities. 
Therefore, the study of systemic risk involves not only understanding the individual components of a network but also examining its pattern of interactions and feedback mechanisms. 
 

 Our main focus in this paper is management of systemic cyber risk, which constitutes a significant part of the general cyber risk exposure, see \cite{Awiszus2023a} for a discussion of classes of cyber risk. Examples of systemic cyber risk incidences are the WannaCry (May 2017) and NotPetya (June 2017) worm-type malware attacks, or the recent IT outage caused by a security update of the cybersecurity firm CrowdStrike (July 2024).

 A common way to understand complex systems is the representation of their components and interactions as \textit{networks}---a term which we use { synonymously} with \textit{graphs}. The nodes of the network represent the individual entities, and edges depict the connections or relationships between them. In cyber systems, these nodes could be computers, servers, or even individual users, while financial systems might have nodes representing banks, markets, or financial investors. Possible interpretations of nodes and edges in electrical systems or in the case of a public transport system are also obvious.
 Networks with multiple layers are applied to represent different classes of interdependent system entities and to model the risk impact of one class on the other.
For illustration consider the potential threat of cyber risks to financial stability. Macroprudential regulators increasingly apply network models to develop a regulatory machinery to control those risks, see e.g.\ \cite{FED2021}, \cite{Euro2020}, \cite{Euro2022}.
In that case the network consists of a cyber (sub)-network and the  financial (sub)-network. So-called ``cyber mappings'' are used to identify systemically relevant network nodes in both the cyber network and the financial network and to analyse transmission channels of cyber risk to the financial system.  A stylized illustration of a cyber map is given in Figure \ref{fig:cybermap}. 
These approaches are discussed, for example, in \cite{Brauchle2020}, \cite{Euro2022}, \cite{Foroutan2024}, and \cite{Schroder2024}.

	\begin{figure}[h]
		\centering
		\includegraphics[width=0.7\textwidth]{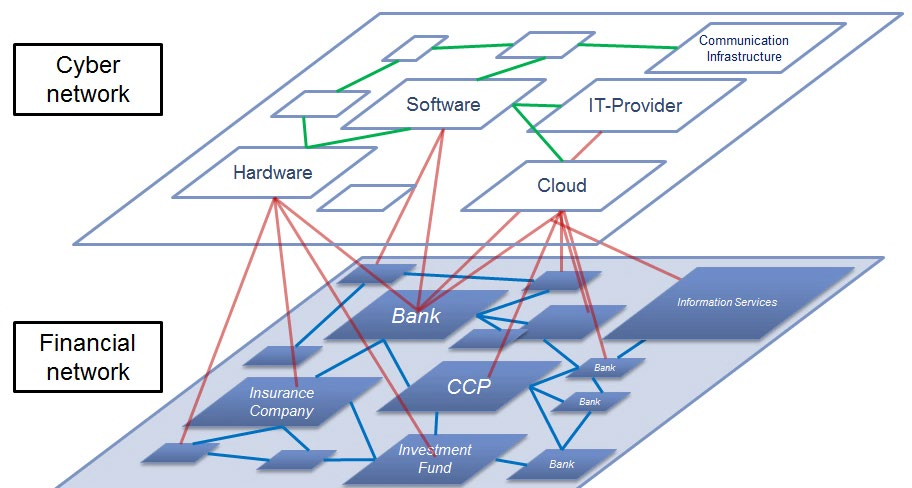}
		\caption{A stylized cyber map consisting of a cyber and a financial (sub-)network. 
		The figure is taken from \cite{Euro2022} and owned by Deutsche Bundesbank. Reproduced with permission.}
	\label{fig:cybermap}
	\end{figure}

The task of a network \textit{supervisor} is to assess the systemic risk of the supervised network and to regulate it in such a way that this risk remains controlled. Network supervisors are, for instance,   
\begin{enumerate}[a)]
    \item a \textit{regulator} or \textit{central planner} managing risk and resilience of (digital) systems and infrastructures from the perspective of society, see the discussion above. 
    
    \item an \textit{insurance company} concerned with the limits of insurability of systemic accumulation risk and the design of policy exclusions.  Indeed insurers act as \textit{private regulators}  by requiring or incentivising security and system standards among their policyholders through contractual obligations, see the discussion in \cite{Awiszus2023b}. 
    \item a \textit{(local) risk manager} or \textit{system administrator} who is concerned with the network protection of a single firm or industry cluster.
\end{enumerate}

In this paper we take the point of view of such a supervisor, and develop a theoretical framework for the assessment and management of systemic (cyber) risk.  To this end, we follow a well-established approach to risk assessment based on three main ingredients:
\begin{itemize}
    \item a set $\mathcal{A}$ of acceptable networks which are deemed to be secure enough,
    \item a set $\mathcal{I}$ of interventions to transform non-acceptable networks into acceptable ones, 
    \item a quantification of the cost $\mathcal{C}$ (not necessarily in monetary terms) of making some network acceptable.
\end{itemize}
The triplet $(\mathcal{A}, \mathcal{I}, \mathcal{C})$ is then called a \textit{decision-making framework for network resilience} (DMFNR). 
In the following we discuss why we identify the three components $\mathcal{A}, \mathcal{I}$, and $\mathcal{C}$ as the central building blocks for assessing systemic risk.

As regards the acceptance set, we find that the most fundamental question any system supervisor has to answer is whether there is a demand for regulation of that system or not. This requires an acceptance criterion as a basis for identifying those networks which the supervisor deems secure enough and those which need to be further secured. Accordingly, the  \textit{acceptance set} contains all secure enough networks. In contrast, networks which do not belong to the acceptance set are those which require further protecting measures. It is up to the supervisor to set the security standards by prescribing the acceptance set. Natural decision criteria in this process are network features such as the presence or absence of systemic parts which may induce concentration risks, or the resilience of the network with respect to (models of) concrete cyber attack scenarios such as the spread of malware (WannaCry, NotPetya) or service outages (CrowdStrike). 

If a given network $G$ does not meet the supervisor's acceptance criterion, that is $G\notin\mathcal{A}$, suitable risk mitigants are needed to secure the network. Those are the supervisor's \textit{interventions} collected in $\mathcal{I}$. In case of cyber risk, well-known examples of risk mitigants in practice are e.g. firewalls, security patches, backups, controlling access to servers or computers, monitoring of data flows, re-configuring network communication and access rules, or diversification of IT-service providers. In our theory, given a network $G$, any measures to improve the network's resilience are interpreted as operations on the given network $G$. The most basic available operations are addition or deletion of nodes and edges. For example, edge hardening measures securing particular network connections (firewalls, increased monitoring of particular network connections, or physical deletions of connections)  
may be identified with the deletion of the corresponding edge from the network, because the increased security erases this risk transmission channel. 
Backups can be modeled by a  combination of node and edge additions, whereas diversification is achieved by node splitting where a node is split into two distinct nodes with the same operational tasks and a redistribution of the original edges. In Section~\ref{sec:topinv} we discuss (theoretical) interventions on the network and their practical interpretations for regulation of systemic cyber risk. Note that a supervisor might typically not have all possible network interventions available, for instance due to legal issues. Therefore, we consider a set of available and admissible network interventions $\mathcal{I}$ from which the supervisor might choose in the regulation process.  
Also note that regulators are increasingly pursuing a ``Zero Trust" principle according to which no one is trusted by default from inside or outside the network. In particular 
no parts of a network are deemed to be more secure than other parts which implies that interactions with other entities of the network cannot be based, for instance, on the security standards of the other. 
As a result, the risk management is left with creating secure and resilient network architectures, see \cite{NIST2020} or \cite[Section 3.5.3]{Nor2023}.



Apart from limitations as regards availability and admissibility of interventions, supervision is also subject to efficiency criteria which we represent by a cost function $\mathcal{C}$. 
We identify and study two main classes of costs, namely costs in terms of monetary investments needed for regulation, and costs in terms of loosing functionality of the network. The latter class of costs is based on well-known measures of network functionality, see Section~\ref{sec:costs}.

 Even though we will mostly limit our discussion to systemic cyber risks,  we want to emphasize that this risk management approach can also be applied to other forms of systemic risk, including those where the trade-off between functionality and risk due to interconnectedness may be different from the cyber risk case. In electrical networks, for instance, an increase in interconnectedness may be desirable from a resilience-building perspective.

\subsection{A Motivating Example}\label{sec:mot}  
Recall the cyber mappings introduced earlier. 
A recent study \cite{Foroutan2024, Schroder2024} provides a detailed cyber mapping analysis of publicly disclosed outsourcing activities by German public retail funds. 
The (layered) network analyzed in this study comprises two main classes of nodes: financial institutions (FIs), representing the financial system, and third-party providers (TPPs), which offer outsourcing services to these FIs. 
We identify two major cyber risks to the financial system that are either caused or exacerbated by outsourcing activities:
\begin{enumerate}
    \item Outages of TPPs, which could severely disrupt the operations of critical parts of the financial system.
    \item Cyberattacks, such as WannaCry or NotPetya, which could spread rapidly through the increased digital interconnections between entities.
\end{enumerate}


Let us discuss different ways a financial system supervisor might decide whether the described network is acceptable or not. One option is to control the number of connections, that is the \textit{degree}, of each node, particularly of the TPP nodes. In the context of risk scenario 1, this strategy ensures that a TPP outage cannot simultaneously disrupt too many connected FIs. For risk scenario 2, controlling the degree helps contain the spread of malware from an infected node to the rest of the system. To manage node degrees effectively, the supervisor could examine the degree distribution of the network. Constraints on the moments of the degree distribution, as a foundation for acceptance sets $\mathcal{A}$, are discussed in detail in Section \ref{sec:degdistrcont}.
On a related note, the supervisor might aim to identify central (systemic) parts of the network that require heightened surveillance. Apart from considering node degrees, there are alternative measures of node centrality which may be useful in this regard, see Section \ref{sec:ncentr}. Acceptance sets based on centrality constraints are discussed in detail in Section~\ref{sec:hub}. Yet a different approach is to simulate  risk scenarios using an external model. The network can then be evaluated based on its performance under such stress tests, see the discussion in Section~\ref{sec:DMFNRstress}.

Regarding the admissible interventions $\mathcal{I}$ of the supervisor, first note that it is reasonable to assume that not every node or connection in the network can be regulated. For example, the European Union's NIS2 Directive \cite{NIS} on the regulation of digital critical infrastructure systems employs a \textnormal{size-cap rule} that limits the directive's scope to entities of medium or large size within the targeted sectors. Appropriate interventions for these nodes might include backup requirements, diversification, or edge hardening measures as discussed above.

In terms of costs, forced diversification or backups on the TPP side are unlikely to affect the overall network functionality, though they may lead to an increase in the monetary cost of outsourced services. In contrast, edge hardening measures could potentially impact network functionality. If costs based on network functionality or true monetary costs do not apply, for instance due to lack of information, a straightforward way to measure costs is to count the number of interventions required to secure the network, aiming to keep this number as low as possible. This approach corresponds to a simplified model of monetary costs, where each intervention is assigned the same symbolic value (e.g., 1 unit). 
However, this model can be refined by assigning different (symbolic) costs to interventions based on factors such as the severity or complexity of the intervention. This would allow for a more nuanced approach to cost modeling, where more invasive or resource-intensive actions incur higher costs. Costs are discussed in Section~\ref{sec:costs}.



Assume now that the regulator has decided on an DMFNR $(\mathcal{A}, \mathcal{I}, \mathcal{C})$. The implementation of this DMFNR can proceed as follows: First, the regulator must assess whether the given network $G$ is acceptable, $G \in \mathcal{A}$,  or if regulation is required, $G \notin \mathcal{A}$. If $G \notin \mathcal{A}$, protective strategies—these are consecutive application of interventions from $\mathcal{I}$—must be found to bring the network into an acceptable state. Suppose that the supervisor has identified various such strategies. Then a reasonable choice amongst those strategies is the most cost-efficient strategy which is determined by means of the cost model $\mathcal{C}$.

\paragraph{Literature}  
As regards the mathematical literature on cyber risk, we refer to the { surveys  \cite{Awiszus2023a} and \cite{eling2020cyber}} for a comprehensive overview. A few recent studies focusing on cyber insurance are \cite{Baldwin2017},  \cite{bessy2020multivariate},  \cite{Boumezoued2023}, \cite{DACOROGNAKRATZ2023}, \cite{Dacorogna2023},  \cite{Hillairet2023}, \cite{zeller2021comprehensive}, \cite{Zeller2023}, and \cite{zellergame2023}.  There is an increasing line of literature {utilizing epidemic processes on networks known from theoretical biology (see, e.g.\,  \cite{Kiss2017,PastorSatorras2015})  to model the spread of contagious cyber risks in digital systems}, see \cite{Antonio2021},  \cite{Chiaradonna2023}, \cite{fahrenwaldt2018pricing}, \cite{Hillairet2021}, \cite{Hillairet2022}, \cite{Jevtic2020}, and \cite{xu2019cybersecurity}. The latter articles are mostly concerned with the pricing of cyber insurance policies.
Further considerations on the significance of systemic cyber risks for the vulnerability of digitally networked societies and economies can be found in \cite{Forscey2022} and \cite{Welburn2022}. 

Systemic risk has long been studied in the context of financial systems, see for instance  \cite{Fouque2013}, \cite{Hurd2016}, and \cite{Jackson2021} for an overview. In particular, models for contagion effects in financial networks have been proposed and studied in \cite{Amini2016},  \cite{Detering2019}, \cite{Detering2019b}, \cite{Eisenberg2001}, and \cite{Gandy2017}, whereas, for example, \cite{Armenti2018},  \cite{Biagini2019}, \cite{Chen2013}, \cite{FoellmerKl2014}, \cite{Foellmer2014}, \cite{Hoffmann2016}, and \cite{Hoffmann2018} consider risk measures for systemic risk in financial markets. 

Various aspects of network resilience and robustness have for example been discussed in \cite{Afrin2018}, \cite{Schaeffer2021} and \cite{Ventresca2014}. In particular, \cite{Afrin2018} emphasizes the need to develop a more general resilience metric for complex networks.

An axiomatic approach to monetary risk measures goes back to \cite{Artzner1999}, see also \cite{Biagini2014}, \cite{FS}, and \cite{Foellmer2015} for an overview of the theory on monetary risk measures. It is from this theory that we adopt the idea to base the risk assessment on a set of acceptable configurations $\mathcal{A}$, a set of controls $\mathcal{I}$, and related costs $\mathcal{C}$.


\paragraph{Outline} Section~\ref{sec:netw} introduces decision-making frameworks for network resilience alongside the formal framework and some foundational definitions. Section~\ref{sec:acceptance} is dedicated to acceptance sets and explores some desirable properties of acceptance sets from the perspective of systemic cyber risk regulation. In Section~\ref{sec:topinv} we discuss theoretical network interventions  $\mathcal{I}$ and their practical counterparts.  Costs $\mathcal{C}$ are introduced and discussed in  Section~\ref{sec:costs}. Examples of DMFNRs are presented in Section \ref{sec:examples}.  Section~\ref{sec:outlook} offers an outlook on open questions and potential areas for future research. Finally, the appendix contains considerations and additional examples for undirected graphs.


\section{Decision-Making Frameworks for Network Resilience}\label{sec:netw}
 In this section we formally define decision-making frameworks for network resilience. To this end, in the following, we first need to specify and briefly discuss their domain of definition. 

\subsection{Graphs and Networks}\label{sec:netdef}
A \textit{graph} $G$ is an ordered pair of sets $G=(\mathcal{V}_G,\mathcal{E}_G)$ where $\mathcal{V}_G$ 
is a finite set of elements, called \textit{nodes} (or \textit{vertices}), and $\mathcal{E}_G\subseteq \mathcal{V}_G\times \mathcal{V}_G$ a set of pairs, called \textit{edges} (or \textit{links}). In the following, we will always assume $(v,v)\notin \mathcal{E}_G$, i.e, we do not allow for self-connections of nodes. Further, for the rest of the paper, we assume that all graphs are defined over the same non-empty countably infinite basis set of nodes $\mathbb{V}$, i.e., that $\mathcal{V}_G\subseteq\mathbb{V}$ for all graphs $G$. We denote the set of such graphs by  $$\mathcal{G}:=\{G\mid G=(\mathcal{V}_G, \mathcal{E}_G),\mathcal{V}_G\subset \mathbb{V},\mathcal{E}_G \subset (\mathcal{V}_G\times \mathcal{V}_G) \cap \mathbb{E} \}$$ where $\mathbb{E}:=\{(v,w)\in \mathbb{V}\times \mathbb{V}\mid v\neq w \}$. The number of nodes $\vert\mathcal{V}_G\vert$ is called the \textit{size} of graph $G$. A graph $G' = (\mathcal{V}', \mathcal{E}')$ with $\mathcal{V}'\subseteq\mathcal{V}_G$, $\mathcal{E}'\subseteq\mathcal{E}_G$ is called a \textit{subgraph} of $G$. 

A \textit{network} is any structure which admits an  abstract representation as a graph: The nodes represent the network's agents or entities, while the connecting edges correspond to a relation or interaction among those entities. In the following, we use the terms ``graph'' and ``network'' interchangeably.

An important subclass of $\mathcal{G}$ is the set of \textit{undirected graphs}: A graph $G= (\mathcal{V}_G, \mathcal{E}_G)\in \mathcal{G}$ is called \textit{undirected} if for all $v,w\in \mathcal{V}_G$ we have $(v,w)\in \mathcal{E}_G \; \Leftrightarrow \; (w,v)\in \mathcal{E}_G$. A graph which is not undirected is often referred to as being \textit{directed}. There is a vast literature on undirected graphs, and many popular network models restrict themselves to this subclass. However, throughout this paper we will work with the complete set of directed graphs $\mathcal{G}$ since the direction of a possible information flow may be essential to understand the associated risk of a corresponding  connection. Nevertheless, our results, or appropriate versions of the results, also hold if we restrict ourselves to the undirected case, see Appendix~\ref{app:undir} for the details. Moreover, in  Appendices ~\ref{sec:epi:thresh} and \ref{sec:specrad} we provide 
examples which are peculiar to an undirected graph framework.  

\begin{example}\label{ex:motdir}
Consider our motivating example from Section~\ref{sec:mot} and the described service outage scenario. In that case the direction of risk propagation is opposite to the direction of outsourcing. Outages originate at the TPP nodes and spread to the FIs that rely on their services. As a result, the network is best viewed as a directed graph with edges running from TPPs to the FIs that outsource to them.
\end{example}

\begin{remark}\label{rem:weighted}(Weighted Graphs)
 In some situations it could be reasonable to consider \textit{weights} assigned to the edges, symbolizing the extent of transmissibility of a shock across each respective edge. Nonetheless, defining the precise nature and magnitude of these weights might pose a challenge. In our motivating example from Section \ref{sec:mot}, only the presence of a connection between an FI and a TPP can be reconstructed from the underlying dataset. There is no data available to meaningfully calibrate potential edge weights in that network. 
 
 Also note that there are cases in which assigning edge weights is not appropriate. Again consider the motivating example. In the event of a TPP outage, all FIs outsourcing to that TPP are equally affected. 
 
 In this paper, we restrict ourselves to directed, but unweighted graphs. We briefly discuss the enhancement of our framework to weighted graphs in the Outlook Section~\ref{sec:outlook}.   
\end{remark}

 In the following, we frequently reference key concepts from network theory, including adjacency matrices, node degrees (both incoming and outgoing), as well as walks and paths between nodes. For the reader not familiar with these concepts, the corresponding definitions and notations are detailed in Appendix~\ref{sec:graph:basics}.

\subsection{Acceptance Sets}

\begin{definition}
An \textit{acceptance set} is any non-empty set $\mathcal{A}\subset \mathcal{G}$.
\end{definition}

The acceptance set is determined by the regulator and comprises all acceptable networks, i.e., the graphs which need no further regulatory interventions. The choice of the correct acceptance set is, of course, a challenging task. We devote Section~\ref{sec:acceptance} to study desirable properties of acceptance sets, and examples are provided in Section~\ref{sec:examples}. Note that, apart from general network characteristics such as the degree distribution of the nodes, acceptability may depend on specific nodes or structure of a concrete network $G\in \mathcal{G}$ which the supervisor is managing. 

\subsection{Interventions}

Interventions are the risk mitigants that are potentially available to the supervisor to secure networks. In general, an intervention can be defined as a mapping that transforms one network into another:
\begin{definition}
A \textit{(network) intervention} is a map $\kappa: \mathcal{G}\to\mathcal{G}$.
\end{definition} 

Examples of interventions are given in Section~\ref{sec:topinv}. Here we also briefly mention, for the case of cyber security, possible interpretations of the presented interventions in terms of real-world measures to protect an IT-system. 

In practice, the supervisor's ability to intervene may be limited. Regulatory approaches often target specific industry sectors and focus on essential network entities, while excluding smaller companies or private users. For instance, the European Union's NIS2 Directive \cite{NIS} on digital critical infrastructure applies a \textnormal{size-cap rule}, restricting its scope to medium and large entities within targeted sectors.
Consequently, risk management is typically constrained to a set of {\em admissible} interventions. Whether an intervention is admissible may depend on the specific graph $G$ that the supervisor must secure.


\begin{definition} Given a non-empty set $\mathcal{I}$ of interventions,  an ($\mathcal{I}$-)\textit{strategy} $\kappa$ is a finite composition of interventions in $\mathcal{I}$, i.e.\ $\kappa = \kappa_1\circ \ldots \circ \kappa_n$ where $\kappa_i\in \mathcal{I}$ for all $i=1,\ldots, n$ and $n\in \mathbb{N}$. We let $[\mathcal{I}]$ denote the set comprising all $\mathcal{I}$-strategies and in addition, if not already included in $\mathcal{I}$, the intervention $\operatorname{id}: G\mapsto G$, leaving the network unaltered. 
The set
\begin{equation*}
    \sigma^{\mathcal{I}} (G) := \{ \kappa (G) \mid  \kappa\in [\mathcal{I}]  \}
\end{equation*}
contains all networks that can be generated from $G$ by applying $\mathcal{I}$-strategies. 
\end{definition}

Of course, the supervisor is constrained to strategies of admissible interventions. The following definition brings together acceptance sets and admissible interventions. 
\begin{definition}\label{def:risk:dec} Let $\mathcal{A}\subset \mathcal{G}$ be an acceptance set and further let $\mathcal{I}$ be a non-empty set of network interventions. $\mathcal{A}$ is said to be \textit{closed} under $\mathcal{I}$ if $\sigma^{\mathcal{I}}(G)\subset \mathcal{A}$ for all $G\in \mathcal{A}$. In that case $\mathcal{I}$ is called \textit{risk-reducing} for $\mathcal{A}$. A network intervention $\kappa$ is called risk-reducing for $\mathcal{A}$  if the intervention set $\{\kappa\}$ is risk-reducing for $\mathcal{A}$.
 \end{definition}

\begin{lemma}\label{lem:partial:order}
Suppose that the non-empty set $\mathcal{I}$ of interventions is not (partially) self-reverse in the sense that for all $G\in\mathcal{G}$ and all $\alpha\in [\mathcal{I}]$ with $\alpha(G)\neq G$, there is no  $\kappa\in[\mathcal{I}]$ such that $ \kappa\circ \alpha(G)=G$. 
    Then $$G\preccurlyeq_{\mathcal{I}} H \quad :\Leftrightarrow \quad H\in \sigma^{\mathcal{I}}(G) $$ defines a partial order on $\mathcal{G}$. In particular, $\mathcal{I}$ is risk-reducing for $\mathcal{A}$ if and only if $\mathcal{A}$  is monotone with respect to $\preccurlyeq_{I}$, that is $G\in \mathcal{A}$ and $G\preccurlyeq_{\mathcal{I}} H $ implies $H\in \mathcal{A}$.
\end{lemma}

\begin{proof}
 $\operatorname{id}\in [\mathcal{I}]$ implies reflexivity of $\preccurlyeq_{\mathcal{I}}$. As regards transitivity, for $H\in\sigma^{\mathcal{I}}(G)$ and an arbitrary $L\in\sigma^{\mathcal{I}}(H)$, we find $\kappa, \tilde{\kappa}\in [\mathcal{I}]$ with $\kappa(G) = H$, and $\tilde{\kappa}(H) = L$, respectively. Then $\tilde{\kappa}\circ\kappa\in [\mathcal{I}]$ and $\tilde{\kappa}\circ\kappa (G) = L$, so that $L\in\sigma^{\mathcal{I}}(G)$. 
     Finally suppose that $G\preccurlyeq_{\mathcal{I}} H$ and $H\preccurlyeq_{\mathcal{I}}  G$, then there are $\mathcal{I}$-strategies $\kappa, \tilde \kappa$ such that $H=\kappa(G)$ and $G=\tilde \kappa(H)=\tilde \kappa \circ \kappa (G)$. Since $\mathcal{I}$ is not partially self-reverse, we must have $H=\kappa(G)=G$. This proves antisymmetry of $\preccurlyeq_\mathcal{I}$. 
\end{proof}


\subsection{Cost Function} 
The final ingredient to our risk management framework is the cost function representing some type of price we pay for altering a given network. The cost function measures the consequences of the supervisors decisions and potentially quantifies the resources needed to achieve acceptability. Therefore, the cost function may depend on the acceptance set and on the interventions applied to achieve acceptability.  
\begin{definition}\label{def:cost}
Given a network acceptance set $\mathcal{A}\subset\mathcal{G}$ and a set of interventions $\mathcal{I}$ which is risk-reducing for $\mathcal{A}$, a map $\mathcal{C}:\mathcal{G}\times\mathcal{G}\to\mathbb{R}\cup\{\infty\}$ with 
\begin{enumerate}
    \item[C1] $\mathcal{C}(G,G) = 0$ for all $G\in\mathcal{G}$
    \item[C2]  $\mathcal{C}(G, H)\geq 0$ whenever $H\in\sigma^{\mathcal{I}}(G)$
\end{enumerate}
is called a \textit{cost function} for $(\mathcal{A}, \mathcal{I})$.
\end{definition}

Property C1 means that not altering the network does not cost anything, while property C2 ensures that any strategy leading to a potential improvement comes at a non-negative cost.

\subsection{Decision-Making Frameworks for Network Resilience}

\begin{definition}
A \textit{decision-making framework for network resilience} (DMFNR) is a triplet $(\mathcal{A}, \mathcal{I}, \mathcal{C})$ where 
\begin{itemize}
\item $\mathcal{A}$ is an acceptance set, 
\item $\mathcal{I}$ is a set of admissible interventions which is risk-reducing for $\mathcal{A}$, 
\item $\mathcal{C}$ is a cost function for $(\mathcal{A}, \mathcal{I})$.
\end{itemize}
\end{definition}

As previously discussed, the acceptance set $\mathcal{A}$ comprises all networks which the supervisor deems secure enough, strategies in $[\mathcal{I}]$ are the risk mitigants of the supervisor in case the network is not acceptable, and the cost function $\mathcal{C}$ determines the cost of securing a network. The requirement that $\mathcal{I}$ be risk-reducing for $\mathcal{A}$ is a minimal consistency between the risk mitigants and the notion of acceptability, meaning that an acceptable network cannot become unacceptable by applying risk mitigants. This could be strengthened, for instance by requirements of the form  $\sigma^{\mathcal{I}}(G)\cap\mathcal{A}\neq \emptyset$ for all networks $G\in\mathcal{G}$. The latter would ensure that any network can be secured, an assumption which we find too strong at this level.  
Property C2 of the cost function is a minimal consistency between strategies securing a network and related costs. 

At this stage the requirements on a DMFNR are rather weak and general. The aim is to provide a framework which we believe is flexible enough, yet structurally reasonable to assess and manage systemic risk. Clearly, a DMFNR is suited for handling systemic risk in many different situations as outlined in the introduction. It is the task of the supervisor to determine a suitable acceptance set, and to identify admissible interventions and related costs to deal with a particular risk situation. The problems here are similar to finding the 'right' risk measure in classical financial regulation (for banks or insurance companies).  In Sections~\ref{sec:acceptance}, \ref{sec:topinv}, and \ref{sec:costs} we will discuss further properties of acceptance sets, interventions, and costs that might be suited to handle certain risk scenarios. The discussion there is, however, mostly tailored to cyber risk in digital networks. In Section~\ref{sec:examples} we provide concrete examples of DMFNRs and we analyse properties of those.   

Note that---in contrast to many approaches to handle risk in the literature---a DMFNR does a priori 
\begin{itemize}
\item[Q1] not attempt to measure the systemic risk in terms of a quantification of risk. In particular, a DMFNR does not answer the question of minimal costs to achieve acceptability (if possible).
\item[Q2] not necessarily allow to compare the risk of two or more networks.
\end{itemize} 
The reason for that is---as is well-known from multivariate risk measure theory in mathematical finance---that a satisfactory answer to Q1 and Q2 might be hard to achieve:
\begin{itemize}
\item As regards Q1 (and also Q2), a reasonable quantification of risk, would be the mentioned minimal costs of achieving acceptability. But that would require complete knowledge of all admissible ways, and the related costs, to transform a network into an acceptable state. For large networks this presents a computability and complexity problem that might be difficult or impossible to solve. Also, depending on the incentives of the supervisor, she might not even be interested in achieving cost efficiency amongst all protection strategies.    \item Dealing with Q2 may require to know the answer to Q1. Eventually it is questionable whether Q2 is really of interest, since in most applications in mind the supervisor is managing a {\em given} network, and is not free to choose an initial configuration to start with.  This given network has to be secured somehow in an admissible and ideally efficient way.
\end{itemize}
A DMFNR sets the framework for the supervisor and can be applied in various ways. As previously emphasized, the acceptance set addresses whether regulation is necessary or not. The admissible risk-reducing interventions include risk mitigation strategies for the supervisor, and the cost function determines the associated costs. If multiple protection strategies are available, the supervisor might opt for the most cost-effective one. However, it may be (computationally) infeasible to determine whether this strategy is globally efficient. Additionally, other factors might be more important than cost efficiency.

However, if we temporarily set aside issues of computability and complexity, 
%
%
a natural way to obtain a (univariate) risk measure based on a DMFNR $(\mathcal{A},\mathcal{I}, \mathcal{C})$ is to let $\rho(G)$ denote the infimal cost of achieving acceptability (if possible):
\begin{equation}\label{eq:uni}
\rho(G) := \inf \{\mathcal{C}(G,\tilde{G}) \mid \tilde{G}\in\sigma^{\mathcal{I}}(G)\cap\mathcal{A}\},\; G\in \mathcal{G}, \quad \mbox{where}\; \inf\emptyset :=\infty. 
\end{equation}

This approach is in line with the construction of monetary risk measures for financial risk management, see \cite{FS} for a comprehensive discussion of monetary risk measures. Constructions of this kind have also been applied to the measurement of systemic financial risks, see for example \cite{Biagini2019}.
\begin{lemma}\label{lem:uni} The univariate measure $\rho$ given in \eqref{eq:uni} satisfies
\begin{enumerate} 
\item $\rho(G)\geq 0$ for $G\notin\mathcal{A}$,
\item $\rho(G) = 0$ whenever $G\in \mathcal{A}$.
\item Suppose, moreover, that there is a minimal cost of alterering a network, that is there exists $\nu>0$ such that $\mathcal{C}(G,\tilde G)\geq\nu$ whenever $G\in \mathcal{G}$ and $\tilde G\in \sigma^{\mathcal{I}}(G)\setminus \{G\}$. Then $\rho (G)\geq \nu$  for all $G\notin\mathcal{A}$ and $\mathcal{A}=\{G\in \mathcal{G}\mid \rho(G)= 0 \}$.
\end{enumerate}
\end{lemma}

\begin{proof}
1.\ immediately follows from property \textit{C2} of the cost function. 2.\ is a consequence of $\operatorname{id}\in [\mathcal{I}]$, thus $G\in\sigma^{\mathcal{I}}(G)$, and of $\mathcal{C}(G, G) = 0$ by property \textit{C1}. 3. is easily verified.
\end{proof}


Finally, note that, apart from complexity issues, in order to apply \eqref{eq:uni} one would have to check whether the infimum is a minimum whenever $\rho(G)<\infty$, that is whether there is a ``best'' acceptable version of the network $G$. 

\section{Properties of Acceptance Sets for Systemic Cyber Risk Management}\label{sec:acceptance}
Consider an acceptance set $\mathcal{A}$. In this section we present selected properties of $\mathcal{A}$ that are tailored to the management of systemic cyber risk.
 
\begin{shaded}
\paragraph{P1} 
\textit{For all $\mathcal{V} \subset \mathbb{V}$ such that $\vert \mathcal{V} \vert \geq 2$, we have $G = (\mathcal{V}, \emptyset) \in \mathcal{A}$.}. 

\hspace{3mm}
\end{shaded}

Property 1 indicates that edgeless graphs should be acceptable due to the absence of contagion channels  or feedback mechanisms. Note the restriction $\vert\mathcal{V}\vert \geq 2$. Indeed, if $\vert\mathcal{V}\vert = 1$, the edgeless graph of size 1 is also a star graph and complete, see below. To maintain consistency with later axioms, we exclude graphs of size 1 from $\mathcal{A}$. This is also justified by observing that occurrence of an incident at the single node $v\in\mathcal{V} = \{ v\}$ means that the full system is  affected, which is worst case from a systemic point of view.

\begin{shaded}
\paragraph{P2} 
\textit{There exists an $N_0 \in \mathbb{N}$ such that for all $N \geq N_0$, the acceptance set $\mathcal{A}$ contains a  weakly connected network with $N$ nodes.}

\hspace{3mm}
\end{shaded}

In view of the functionality of a system, interconnectedness should be acceptable up to a certain degree. We refer the reader not familiar with the notion of weakly and strongly connected networks to Appendix~\ref{sec:connectivity}. The size restriction in P2 accounts for the observation that with an increasing number of nodes, maintaining functionality, that is the network being connected, while controlling the risk is more feasible.





Next we address system vulnerability by specifying worst-case arrangements of networks that exceed any reasonable limit of risk. For the notion of in- and out-neighborhoods of nodes $v\in \mathcal{V}_G$, denoted by $\mathcal{N}_{v}^{G, in}$ and $\mathcal{N}_{v}^{G, out}$, respectively,  we refer to Appendix \ref{sec:neighbor}. 

\begin{definition}
Consider a network $G\in\mathcal{G}$ and a node $v^\ast\in\mathcal{V}_G$.
\begin{enumerate}
\item We call $v^\ast$ a \textit{super-spreader} if $\mathcal{N}_{v^\ast}^{G, out} = \mathcal{V}_G\setminus\{v^\ast\}$.
\item If, moreover, $\mathcal{N}_{v^\ast}^{G, in} = \mathcal{N}_{v^\ast}^{G, out} = \mathcal{V}_G\setminus\{v^\ast\}$, then $v^\ast$ is called a \textit{star node}.
\end{enumerate}
\end{definition}

\hspace{3mm}

\begin{figure}[h]
\begin{center}
			
		\begin{minipage}[t]{0.49\linewidth}
			\centering

			{\includegraphics[trim={1.2cm 5.5cm 0cm 4cm},clip,width=0.65\textwidth]{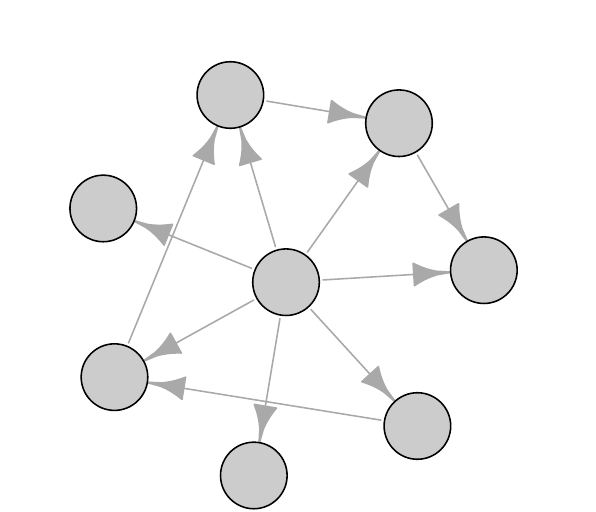}}
		\end{minipage}
  \begin{minipage}[t]{0.49\linewidth}
			\centering
            {\includegraphics[trim={1.2cm 5.5cm 0cm 4cm},clip,width=0.65\textwidth]{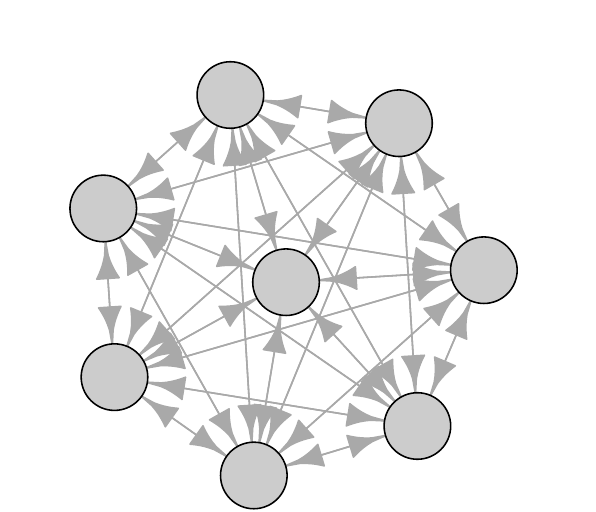}}	
		\end{minipage}
		\end{center}
		\caption{A graph that contains a super-spreader node (left), and a complete graph, where any node is a star node.}
\end{figure}


\begin{shaded}
\paragraph{P3}  \textit{$\mathcal{A}$ does not contain any network with a super-spreader}.

\hspace{3mm}
\end{shaded}

P3 aims to exclude networks where the infection or outage of a single node may affect all others.  Note that P3 also excludes all networks with a star node, which have the additional disadvantage that they themselves may be infected by any other node, and then potentially pass on that infection. Specifically, this rules out any \textit{complete graph} $G^c = (\mathcal{V}_{G^c}, (\mathcal{V}_{G^c}\times\mathcal{V}_{G^c})\cap\mathbb{E})$, where each node is effectively a star node. 
If we do not endow the single nodes $v\in\mathcal{V}_G$ of a network $G$ with a particular meaning, one might identify networks which are structurally equal: Two graphs $G_1 = (\mathcal{V}_1, \mathcal{E}_1)$ and $G_2=(\mathcal{V}_2, \mathcal{E}_2) \in\mathcal{G}$ are isomorphic if there is a bijection $\pi: \mathcal{V}_1\to\mathcal{V}_2$ such that
\begin{equation*}\label{eq:iso}
    (v,w)\in\mathcal{E}_{1} \Leftrightarrow (\pi (v), \pi (w))\in\mathcal{E}_{2} \text{ for all } v, w \in\mathcal{V}_1.
\end{equation*}
$\pi$ is called a \textit{graph isomorphism} of $G_1$ and $G_2$. 

\begin{shaded}
\paragraph{P4} \textit{Topological invariance: $\mathcal{A}$ is closed under graph isomorphisms, i.e., for any two networks $G, \tilde{G}\in\mathcal{G}$ which are isomorphic, we have $G\in\mathcal{A}$ if and only if $\tilde{G}\in\mathcal{A}$.}

\hspace{3mm}
\end{shaded}

P4 resembles the concept of distribution-invariant  risk measures common in financial economics. Indeed, if graphs are isomorphic, then their degree distributions coincide. Note that in scenarios where network components represent specific entities, some of which may need to be prioritized for protection---such as those belonging to critical infrastructures---topological invariance is not suitable. For instance, topology-invariant acceptance sets may not be appropriate in our motivating example from Section~\ref{sec:mot}, where we fundamentally distinguish between TPP and FI nodes. 

\section{Network Interventions}\label{sec:topinv}

In this section, we present various types of interventions, beginning with elementary ones. Some of these manipulations have been previously applied in the literature on network risk management; see, for example, \cite{Awiszus2023b} and \cite{Chernikova2022} for the case of cyber networks. In \cite{Louzada2013}, edge shifts are evaluated for enhancing the robustness and resilience of air transportation networks. Throughout this section, \(G = (\mathcal{V}_G, \mathcal{E}_G) \in \mathcal{G}\) denotes a generic graph.

\subsection{Elementary Interventions on Edges and Nodes}\label{sec:elem}

\begin{itemize}
    \item[$\mathcal{I}_{e\_ del}$] \textit{Edge Deletion}: Consider nodes $v,w\in\mathbb{V}$. The intervention 
\begin{equation*}
    \kappa^{(v,w)}_{e\_ del}:
    G\mapsto 
    (\mathcal{V}_G, \mathcal{E}_G\setminus\{ (v,w)\})
\end{equation*}
deletes the edge $(v,w)$ if present. Indeed $\kappa^{(v,w)}_{e\_ del} (G) = G$ if and only if $(v,w)\notin\mathcal{E}_G$. We set $$\mathcal{I}_{e\_ del}:=\{\kappa^{(v,w)}_{e\_ del}\mid v,w\in\mathbb{V}\}.$$ 
Operational interpretations of edge deletion for controlling cyber risk are 
   physical deletion of connections, such as access to servers, or, if not possible, edge hardening, i.e., a strong protection of network connections via firewalls, the closing of open ports, or the monitoring of data flows using specific detection systems.

\item[$\mathcal{I}_{e\_ add}$] \textit{Edge Addition}: The addition of edge $(v,w)$ for $v,w\in\mathbb{V}$ with $v\neq w$ is given by 

\begin{equation*}
    \kappa^{(v,w)}_{e\_ add}:
    G\mapsto 
    (\mathcal{V}_G, \mathcal{E}_G\cup (\{ (v,w)\}\cap \mathcal{V}_G\times\mathcal{V}_G)),
\end{equation*}
where $ \kappa^{(v,w)}_{e\_ add}(G) = G$ is satisfied if either we already have $(v,w)\in\mathcal{E}_G$ or at least one of the nodes $v$ and $w$ is not present in the network $G$. We set $$\mathcal{I}_{e\_ add}:=\{\kappa^{(v,w)}_{e\_ add}\mid v,w\in\mathbb{V}\}.$$ Edge addition is, of course, the reverse of edge deletion and comes with the opposite operational meaning. In the context of electrical networks, edge addition may be interpreted as physical addition of power lines between nodes, thereby securing power supply.

\item[$\mathcal{I}_{n\_ del}$] \textit{Node Deletion}: The deletion of a node $v\in\mathbb{V}$ is given by the intervention
\begin{equation*}
    \kappa^{v}_{n\_ del}:
     G\mapsto 
     (\mathcal{V}_G\setminus\{v\}, \mathcal{E}_G\cap (\mathcal{V}_G\setminus\{v\}\times \mathcal{V}_G\setminus\{v\}))
\end{equation*}
where $\kappa^{v}_{n\_ del} (G) = G$ if and only if $v\notin\mathcal{V}_G$. We set $$\mathcal{I}_{n\_ del}:=\{\kappa^{v}_{n\_ del}\mid v\in\mathbb{V}\}.$$
Operationally node deletion corresponds, for instance, to removal of redundant servers or other access points. 

\item[$\mathcal{I}_{n\_ add}$] \textit{Node Addition}: An (isolated) node $v\in\mathbb{V}$ can be added to a network via
\begin{equation*}
    \kappa^v_{n\_ add} : G\mapsto (\mathcal{V}_G\cup\{v\}, \mathcal{E}_G)
\end{equation*}
where $\kappa^v_{n\_ add}(G) = G$ if and only if we already have $v\in\mathcal{V}_G$. We set $$\mathcal{I}_{n\_ add}:=\{\kappa^{v}_{n\_ add}\mid v\in\mathbb{V}\}.$$ Node addition, such as adding servers that serve as backups or take over tasks from very central servers, may improve network resilience to the spread of infectious malware when combined with edge addition or edge shifts (see below).
\end{itemize}

More complex interventions such as discussed in the following can be obtained by a consecutive application of the elementary interventions presented so far.   

\subsection{Isolation of Network Parts}
Critical nodes or subgraphs can be isolated from the rest of the network by application of the following edge deletion procedures:
\begin{itemize}
\item[$\mathcal{I}_{s\_ iso}$] \textit{Node/Subgraph Isolation}: We isolate the subgraph $(\mathcal{W}\cap \mathcal{V}_G, \mathcal{E}_G\cap (\mathcal{W}\times \mathcal{W}))$ given by a set $\mathcal{W}\subset\mathbb{V}$ (or a node in case $\mathcal{W}=\{v\}$ for some $v\in \mathbb{V}$) from the rest of the network by 
\begin{equation*}
    \kappa^{\mathcal{W}}_{iso}: G\mapsto  (\mathcal{V}_G,  \mathcal{E}_G\setminus [\mathcal{W}\times (\mathcal{V}_G\setminus\mathcal{W}) \cup  (\mathcal{V}_G\setminus\mathcal{W})\times\mathcal{W}]), \quad \mathcal{I}_{s\_ iso}:=\{\kappa^\mathcal{W}_{iso}\mid \mathcal{W}\subset \mathbb{V}\}.\end{equation*}
Node or subgraph isolation corresponds to fully protecting that part of the network from the rest.
\end{itemize}


\subsection{Shifting, Diversification and Centralization of Network Connections} 
Network tasks or data flows can be redistributed among existing nodes by utilizing interventions of the following type:
\begin{itemize}
    \item[$\mathcal{I}_{shift}$] \textit{Edge Shift}: If $(v,w)\in\mathcal{E}_G$ and $(q,r)\in\mathbb{E}\setminus\mathcal{E}_G$, $q,r\in \mathcal{V}_G$, then the edge $(v,w)$ can be shifted to $(q,r)$ by
\begin{equation*}
    \kappa_{shift}^{(v,w), (q,r)}:  G\mapsto \begin{cases} (\mathcal{V}_G, \mathcal{E}_G\setminus\{ (v,w)\} \cup \{(q,r)\}), &\textit{if } (v,w)\in\mathcal{E}_G, q,r\in \mathcal{V}_G, (q,r)\in\mathbb{E}\setminus\mathcal{E}_G \\
    G, &\textit{else.}\end{cases}
\end{equation*} We set $$\mathcal{I}_{shift}:=\{\kappa_{shift}^{(v,w), (q,r)}\mid v,w,q,r\in \mathbb{V}\}.$$
\end{itemize}
Combining node adding and edge shift operations, we can create more complex interventions which aim for risk diversification by separating critical contagion channels that pass through a chosen node $v$:  

\begin{itemize}
    \item[$\mathcal{I}_{split}$] \textit{Node Splitting}: The split of a node $v$ by rewiring edges contained in $\mathcal{L}\subset \mathbb{E}$ from $v$ to a newly added node $\tilde{v}\in \mathbb{V}$ is described by 
\begin{equation*}
 \kappa_{split}^{\mathcal{L}, v, \tilde{v}}: G\mapsto \begin{cases} (\mathcal{V}_G\cup\{\tilde{v}\}, \mathcal{E}_G\setminus (\mathcal{L}\cap\mathcal{E}_G^v) \cup \tilde{\mathcal{L}}^G_{\tilde v}), & \textit{if } \tilde{v}\notin\mathcal{V}_G\\
 G, &\textit{else,}\end{cases}
\end{equation*}
with $\tilde{\mathcal{L}}^{G}_{\tilde v} := \{(\tilde{v}, w)\vert (v, w)\in\mathcal{L}\cap\mathcal{E}_G\}\cup\{(w, \tilde{v})\vert (w, v)\in\mathcal{L}\cap\mathcal{E}_G\}$. A true node splitting only takes place when $\tilde v\not \in \mathcal{V}_G$ and $\mathcal{L}\cap\mathcal{E}_G^v\neq\emptyset$. Let $$\mathcal{I}_{split}:=\{\kappa_{split}^{\mathcal{L}, v, \tilde{v}}\mid v,\tilde v\in \mathbb{V}, \mathcal{L}\subset \mathbb{E}\}.$$
Note that node splitting enters the practical discussion when, for instance, a required diversification of (outsourcing) service providers is considered by the regulator. In view of the motivating example from Section~\ref{sec:mot} such an intervention could reduce the impact of an outage or infection of a TTP to the financial system. 
\end{itemize}
Node splits can be reversed by \textit{merging} nodes $v$ and $\tilde{v}$. This operation corresponds to a centralization of network connections, the formalism is left to the reader.
\begin{figure}[h]
\begin{center}
\begin{minipage}[t]{0.35\linewidth}
\centering
	{\includegraphics[trim = 20mm 8.5mm 12mm 12mm, clip, width=1.2\textwidth]{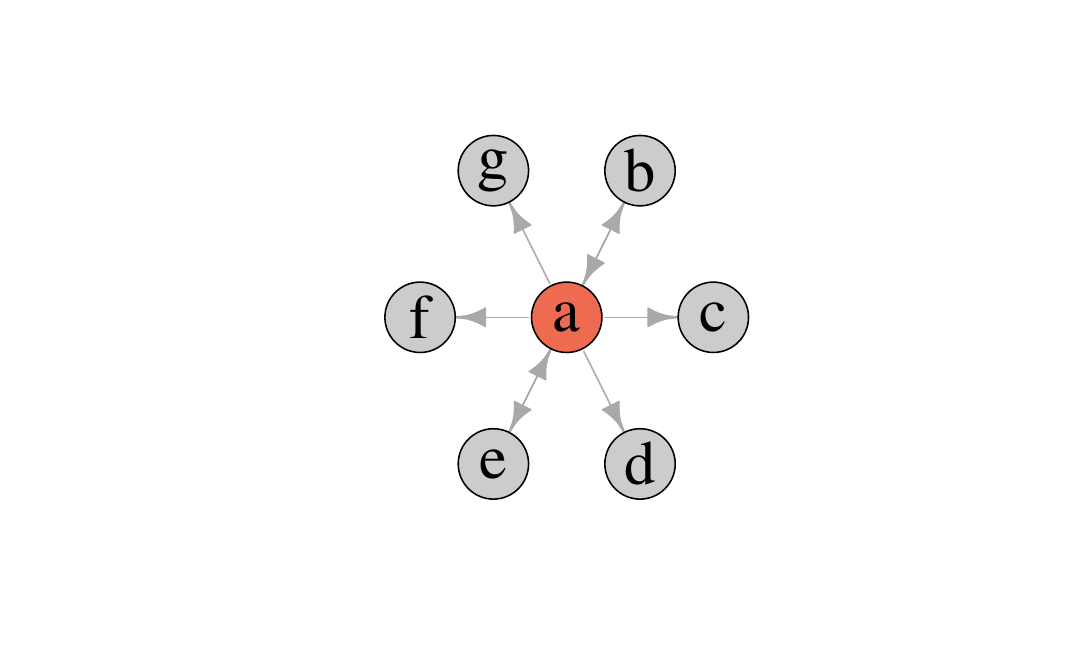}}
\end{minipage}
\begin{minipage}[t]{0.15\linewidth}
			\centering
			\vspace{-2.8cm}
			\begin{tikzpicture}[scale=0.9] 
			\tikzset{edge/.style = {->,> = latex'}}
\draw[arrow, line width=0.7mm] (-2,2) --  node[anchor=south] {split}   (0.5,2) ;
\draw[arrow, line width=0.7mm] (0.5, 1.4) -- node[anchor=north] {merge}  (-2,1.4) ;


\end{tikzpicture}
\end{minipage}
\begin{minipage}[t]{0.35\linewidth}
	{\includegraphics[trim = 35mm 5mm 12mm 20mm, clip, width=0.96\textwidth]{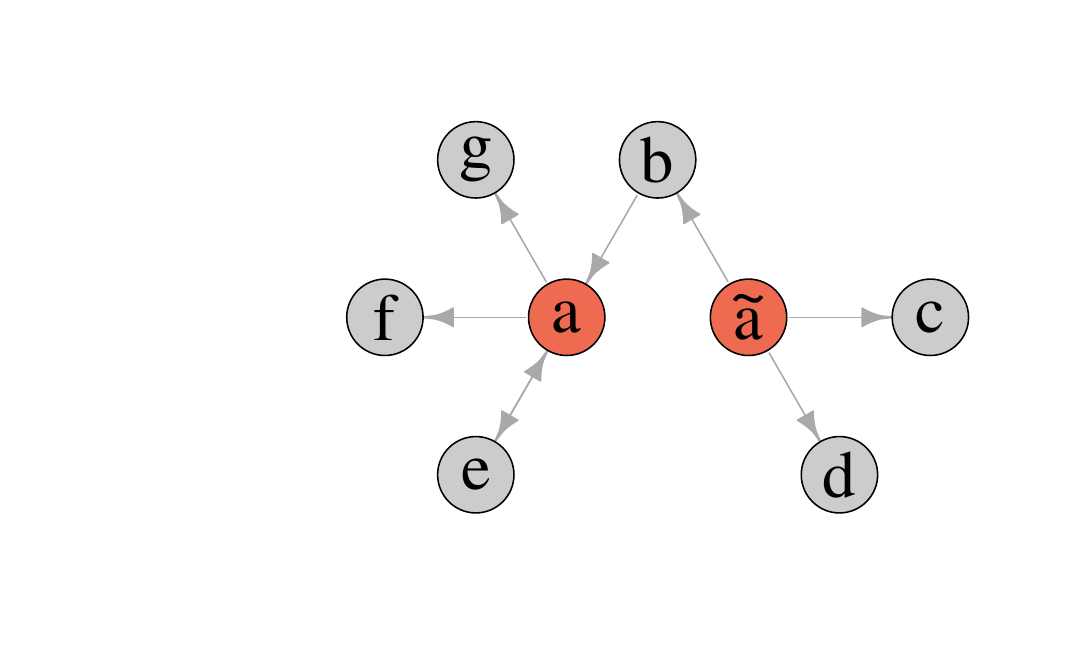}}
	\end{minipage}
\begin{minipage}{0.09\linewidth}
\hspace{15mm}
\end{minipage}
\end{center}
\vspace{-10mm}
\caption{Splitting (left to right) and merging (right to left) as mutually reverse operations. Targeted nodes are coloured in red.}		
\end{figure}
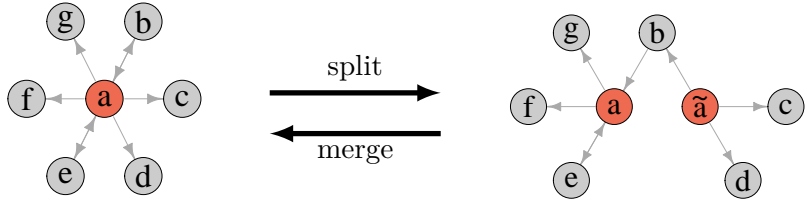


Creating node copies has an obvious practical interpretation as installing backups: 
\begin{itemize}
\item[$\mathcal{I}_{n\_ copy}$]  \textit{Node Copying:} A node $v\in\mathcal{V}_G$ can be copied by 
\begin{equation*}
\kappa^{v, w}_{copy}: G\mapsto \begin{cases} (\mathcal{V}_G\cup\{ w\},\textit{}\mathcal{E}_G\cup\{ (u, w)\vert (u,v)\in\mathcal{E}_G\}\cup \{ (w, u)\vert (v, u)\in\mathcal{E}_G\}), &\textit{ if } v\in\mathcal{V}_G, w\notin\mathcal{V}_G  \\
G, &\textit{else.}\end{cases}
\end{equation*}
Note that a copy of node $v$ can be stored at node $w$ only if $v\in\mathcal{V}_G$ and $w\notin\mathcal{V}_G$. We set $$\mathcal{I}_{n\_ copy}:=\{\kappa^{v, w}_{copy}\mid v,w\in \mathbb{V}\}.$$
\end{itemize}

\section{Cost Functions}\label{sec:costs}

The following two types of cost functions appear naturally:
\begin{enumerate}
    \item A quantification of the \textit{monetary} cost of transforming a graph $G$ into $H\in \sigma^{\mathcal{I}}(G)$, see Section~\ref{sec:monetary:costs}. 
    \item A quantification of the cost of transforming a graph $G$ into $H\in \sigma^{\mathcal{I}}(G)$ in terms of the corresponding  \textit{impact on network functionality}, see Section~\ref{sec:net:func}.
\end{enumerate}

In the case studies in Sections~\ref{sec:monetary:costs} and \ref{sec:net:func} we will also consider further properties of a cost function which we introduce next. To this end, let $\mathcal{A}$ be an acceptance set, $\mathcal{I}$ be a set of risk reducing interventions for $\mathcal{A}$, and suppose that $\mathcal{C}$ is a cost function for $(\mathcal{A}, \mathcal{I})$.
\begin{enumerate}
\item[C3] Every non-trivial network intervention comes at a real cost, i.e., \[\mathcal{C}(G,H) = 0\Rightarrow G=H \text{ for all } H, G\in\mathcal{G}.\]
\item[C4] The absolute cost of a direct transformation of $G$ to $H$ should not be more expensive than first transforming to an intermediate configuration $M$, i.e., for all $G, H,M\in\mathcal{G}$ such that $M\in\sigma^{\mathcal{I}}(G)$, and $H\in\sigma^{\mathcal{I}}(M)$ 
we have 
\[\mathcal{C}(G, H) \leq \mathcal{C}(G,M) + \mathcal{C}(M, H).\]
\end{enumerate}

\subsection{Monetary Cost Functions}\label{sec:monetary:costs}
Consider any acceptance set \(\mathcal{A} \subset \mathcal{G}\) and a set of interventions \(\mathcal{I}\) that is risk-reducing for \(\mathcal{A}\). In the following, we illustrate monetary cost functions by presenting an example of a cost function constructed based on the costs of the basic interventions \(\kappa \in \mathcal{I}\). These costs are assumed to be initially given.

Suppose we associate an initial cost $\tilde p(\kappa)\in [0,\infty)$ with each 
    $\kappa\in \mathcal{I}$ and let $\tilde p(\operatorname{id})=0$. For any $\mathcal{I}$-strategy $\kappa$ set \begin{equation}\label{eq:p} p(\kappa) :=\inf \left\{ \sum_{i=1}^n \tilde p(\kappa_i)\mid \kappa=\kappa_n\circ\cdots\circ\kappa_1, \kappa_i\in \mathcal{I}, i=1,\ldots, n, n\in \mathbb{N} \right\}.\end{equation} Then consider the minimal total cost for the transformation of $G$ into $H$: 
    \begin{equation}\label{eq:costmon}
    \mathcal{C}(G,H):=\inf\{p(\kappa)\mid \kappa \in [\mathcal{I}],  \kappa(G)=H\},  \quad \inf\emptyset:=\infty.
    \end{equation}
Note that $p$ conincides with $\tilde p$ on $\mathcal{I}$ if and only  if  $\tilde p$ is {\em consistent}, that is if $\kappa\in \mathcal{I}$ satisfies $\kappa=\kappa_n\circ\cdots\circ\kappa_1$ for some interventions $\kappa_i\in \mathcal{I}$, $i=1,\ldots, n$, $n\in \mathbb{N}$, we must have $ \tilde p(\kappa)=\sum_{i=1}^n \tilde p(\kappa_i)$. Of course, consistency of $\tilde p$ on $\mathcal{I}$ is automatically satisfied in case that  
$\kappa=\kappa_n\circ\cdots\circ\kappa_1$ for some $\kappa, \kappa_1,\ldots, \kappa_n\in \mathcal{I}$ implies that $\kappa_j=\kappa$ for some $j\in \{1,\ldots, n\}$ and $\kappa_i=\operatorname{id}$ for all $i\neq j$. The following lemma is easily verified. 

\begin{lemma}\label{lem:p}
$p$ is subadditive in the sense that for all $\kappa, \alpha \in [\mathcal{I}]$ we have  $p(\kappa\circ\alpha)\leq p(\kappa) + p(\alpha)$.
\end{lemma}


\begin{proposition}\label{prop:mon:costs} 
Suppose that  $\mathcal{C}$ is given by  \eqref{eq:costmon}.
Then: 
\begin{enumerate}
\item $\mathcal{C}(G,H)<\infty$ if and only if $H\in \sigma^{\mathcal{I}}(G)$.
\item $\mathcal{C}$  is a cost function for $(\mathcal{A}, \mathcal{I})$ which satisfies \textit{C4}.  
\item Further suppose that $\inf\{ \tilde p(\kappa)\mid \kappa \in \mathcal{I}\setminus \{\operatorname{id} \}\}>0 $, then \textit{C3} is satisfied. 

\item Moreover, if $\mathcal{I}$ is finite and $\tilde p(\kappa)>0$ for all $\kappa\in \mathcal{I}$ such that $\kappa\neq \operatorname{id}$, then the infimum in \eqref{eq:costmon} is attained whenever $H\in \sigma^{\mathcal{I}}(G)$. In other words, if $H\in \sigma^{\mathcal{I}}(G)$, then there are optimal $\mathcal{I}$-strategies $\kappa$  for transforming $G$ into $H$ in the sense that $\kappa(G)=H$ and $p(\kappa)=\mathcal{C}(G,H)$.  
\end{enumerate}
\end{proposition} 
 
 \begin{proof}
 1. is obvious.
 
 \smallskip\noindent
 2. $\mathcal{C}$ is a cost function for $(\mathcal{A}, \mathcal{I})$ which satisfies \textit{C4}: \textit{C1} follows from $p(\operatorname{id})=0$, whereas \textit{C2} is implied by $\tilde p(\kappa)\geq 0$  for each $\kappa\in \mathcal{I}$. As for \textit{C4} consider $G,M,H\in \mathcal{G}$ and suppose that there is $\kappa, \alpha \in [\mathcal{I}]$ with $\alpha (G)=M$ and $\kappa(M)=H$.  Then $\kappa\circ \alpha(G)=H$, and  by Lemma~\ref{lem:p} above we have that $\mathcal{C}(G,H) \leq \mathcal{C}(G,M) + \mathcal{C}(M,H)$.
 
 \smallskip\noindent
 3. Clearly, if $\nu:=\inf\{ p(\kappa)\mid \kappa \in \mathcal{I}\setminus \{\operatorname{id} \}\}>0 $, then $p(\kappa)>\nu$ for all $\kappa\in [\mathcal{I}^\downarrow]$ such that $\kappa\neq \operatorname{id}$. Hence, $\mathcal{C}(G,H)=0$ is only possible if $H=G$. 
 
 \smallskip\noindent
4. In order to show that there are optimal $\mathcal{I}$-strategies in case $H\in \sigma^{\mathcal{I}}(G)$, pick any $\kappa\in [\mathcal{I}]$ such that $\kappa(G)=H$.  Note that $\nu=\min \{\tilde p(\kappa)\mid \kappa\in \mathcal{I}\setminus\{\operatorname{id}\}\}>0$, and let $m\in \mathbb{N}$ be such that $m\nu > p(\kappa)$. Denote by $[\mathcal{I}]_m$ the set of all $\mathcal{I}$-strategies which allow for a decomposition $\alpha=\alpha_k\circ \ldots \alpha_1$ where $\alpha_1,\ldots, \alpha_k\in \mathcal{I}$ and $k\leq m$. Note that $p(\alpha)\geq m\nu >p(\kappa)$ for all $\alpha\in [\mathcal{I}]\setminus [\mathcal{I}]_m$. Therefore, it follows that  $$\mathcal{C}(G,H)=\inf\{p(\alpha)\mid \alpha \in [\mathcal{I}]_m,  \alpha(G)=H\}.$$ Now the assertion follows from $[\mathcal{I}]\setminus [\mathcal{I}]_m$ being a finite set.
  \end{proof}

\subsection{Examples of Cost Functions based on Loss of Network Functionality}\label{sec:net:func}
In this section we consider cost functions which quantify the difference in functionality $\mathcal{F}:\mathcal{G}\to\mathbb{R}$ between two networks:
\begin{equation}\label{eq:cost}
     \mathcal{C}(G, H) = h( \mathcal{F}(G) - \mathcal{F}(H))
\end{equation}
where $h:\mathbb{R}\to\mathbb{R}$. Note that the Properties \textit{C1} and \textit{C4} solely depend on the choice of $h$, and one easily verifies the following result:
\begin{lemma} Consider a cost function of type \eqref{eq:cost}. Then 
\begin{enumerate}
    \item Property \textit{C1} is satisfied if and only if $h(0) = 0$.
    \item If $h$ is sub-additive, then \textit{C4} holds.
\end{enumerate}
\end{lemma}

Network functionality is typically measured by considering the average node distance within the network: A smaller average value corresponds to a faster or more efficient data flow. In following we recall two common concepts to define the distance between nodes. To this end, we refer the reader not familiar with basic notions for graphs such as paths and adjacency matrices to Appendix~\ref{sec:graph:basics}.

\begin{itemize}
    \item \textit{Shortest path length:}  Given two nodes $v$ and $w$ in a network $G$, we denote by $l_{vw}^G$ the length of a path from node $v$ to $w$ that passes through a minimum number of edges. If there is no path from $v$ to $w$, then we set $l_{vw}^G = \min\emptyset :=\infty$. However, note that it has been emphasized, see e.g.\ \cite{Estrada2012}, that considering shortest paths only provides an incomplete picture of the flow of information, data, or goods in a network since that may take place through many different routes. 
    

    \item \textit{Communicability:} Compared to shortest path length, \textit{communicability} attempts to provide a more holistic view of the communication channels between nodes. Consider a graph \(G\) of size \(N\) with enumerated vertices \(v_1, \ldots, v_N\). The communicability \(\mathfrak{c}_{v_i v_j}^G\) between nodes \(v_i\) and \(v_j\) is defined as the \((i,j)\)-th entry of the matrix exponential \(e^{A_G}\) of the adjacency matrix \(A_G\), i.e.,
\begin{equation*}
    \mathfrak{c}_{v_iv_j}^G := \sum_{k=0}^\infty \frac{(A_G^k) (i,j)}{k!} = e^{A_G} (i,j),\; \mbox{where}\; e^{A_G} := I_N + A_G + \frac{A^2_G}{2!}+\frac{A_G^3}{3!}+\cdots = \sum_{k=0}^\infty \frac{A_G^k}{k!}
\end{equation*}
and $I_N$ is the identity matrix in $\mathbb{R}^{N\times N}$. Note that  $A_G^k (i,j)$ gives the number of walks of length $k$ from node $v_i$ to $v_j$, and therefore $\mathfrak{c}_{v_iv_j}^G$ is a weighted sum of all walks from $v_i$ to $v_j$. 
\end{itemize}


\subsubsection{Functionality Based on Shortest Paths} A global measure of network functionality based on shortest paths is given by the \textit{graph efficiency}  as defined in \cite{Latora2001}:
		\begin{equation}\label{eq:greff}
			\mathcal{F}(G) =  \frac{1}{|\mathcal{V}_G|(|\mathcal{V}_G|-1)} \sum_{v,w, v\neq w}\frac{1}{l^G_{vw}} \qquad \Big(\frac{1}{\infty}:=0\Big).
		\end{equation}

{ 
\begin{proposition}\label{prop:edelshort}
Consider an acceptance set $\mathcal{A}$ and a set of admissible interventions $\mathcal{I}$ which is risk-reducing for $\mathcal{A}$ such that  $\mathcal{I}\subset\mathcal{I}_{e\_ del} \cup   \mathcal{I}_{s\_iso}$. Then $\mathcal{C}$ given by \eqref{eq:cost} and \eqref{eq:greff} and a strictly increasing function $h$ with $h(0) = 0$ is a cost function for $(\mathcal{A},\mathcal{I})$. Moreover, $\mathcal{C}$ satisfies the following constrained version of \textit{C3}: 
\begin{equation}\label{eq:C3}
\forall G \in \mathcal{G}, \, \forall H \in \sigma^{\mathcal{I}}(G): \quad \mathcal{C}(G, H) = 0 \, \Rightarrow \, G = H.
\end{equation}
\end{proposition}
\begin{proof}
\textit{C1} follows from $h(0)=0$. In order to prove \textit{C2} and \eqref{eq:C3}, we show that $\mathcal{F}(G)$ is non-increasing whenever we apply a intervention in $\mathcal{I}_{e\_ del} \cup   \mathcal{I}_{s\_iso}$ to $G$, and even strictly decreasing if the intervention alters $G$. As a consequence, for all \(\kappa \in [\mathcal{I}_{e\_ \text{del}} \cup \mathcal{I}_{s\_ \text{iso}}]\), and thus for all \(\kappa \in [\mathcal{I}]\), we obtain that \(\mathcal{F}(\kappa(G)) \leq \mathcal{F}(G)\), which implies \(h(\mathcal{F}(G) - \mathcal{F}(\kappa(G))) \geq h(0) = 0\), demonstrating \textit{C2}. Additionally, \(\mathcal{F}(\kappa(G)) < \mathcal{F}(G)\) when \(\kappa(G) \neq G\), implying \(h(\mathcal{F}(G) - \mathcal{F}(\kappa(G))) > h(0) = 0\), which corresponds to \eqref{eq:C3}.

Let us first consider edge deletions: Fix $G\in\mathcal{G}$ and delete an edge $(q,r)\in\mathbb{E}$. Let $\tilde{G} := \kappa^{(q,r)}_{e\_ del}(G)$. The deletion of the edge $(q,r)$ does not create any new paths, and therefore $l_{vw}^G\leq l_{vw}^{\tilde{G}}$ for all nodes $v,w\in\mathcal{V}_G$.  Moreover, if $(q,r)\in\mathcal{E}_G$, then we necessarily have $l_{qr}^G < l_{qr}^{\tilde{G}}$, i.e., there is at least one node pair where the inequality is strict. We thus obtain $\mathcal{F}(\tilde{G})\leq\mathcal{F}(G)$ for all networks $G\in\mathcal{G}$ and $(q,r)\in\mathbb{E}$, and ``='' holds exactly when we have $(q,r)\not\in\mathcal{E}_G$, i.e., if and only if $\tilde{G}= G$. 

Note that the isolation of a subgraph, that is  
$\tilde{G}:=\kappa^{\mathcal{W}}_{iso}(G)$ where $\mathcal{W}\subset\mathbb{V}$, is a sequence of edge deletions, namely deleting all edges connecting nodes in $\mathcal{V}_G\setminus \mathcal{W}$ with some node in $\mathcal{W}$. Hence, also in this case we obtain $\mathcal{F}(\tilde{G})\leq\mathcal{F}(G)$, and ``='' holds exactly when $v\not\in \mathcal{V}_G$ for all $v\in \mathcal{W}$ or when the subgraph given by the nodes $\mathcal{W}\cap \mathcal{V}_G$ is already in  $G$, i.e., if and only if $\tilde{G}= G$.

%
\end{proof}}

However, measuring network functionality in terms of shortest paths may come with some difficulties when allowing for $\mathcal{I}_{\text{n}\_split}\subseteq\mathcal{I}$ as is illustrated by the following example.

\begin{example} \label{ex:nonsplit}
 Let $H$ be the result of a node split in a network $G$ of size $N$. For $\mathcal{F}$ to be monotonically decreasing under this intervention, the shortest path lengths in $H$ must satisfy:
\begin{equation} \label{eq:nonsplit}
    \frac{N-1}{N+1} \sum_{v,w\in\mathcal{V}_H, v\neq w}\frac{1}{l^H_{vw}} \leq \sum_{v,w\in\mathcal{V}_G, v\neq w}\frac{1}{l^G_{vw}}.
\end{equation}

Consider a graph $G=G_1\cup G_2$ consisting of two components, where $G_1$ is as shown in Figure \ref{fig:nonsplit}, and $G_2$ is a graph of $N-3$ isolated nodes such that $\mathcal{V}_{G_1} \cap \mathcal{V}_{G_2} = \emptyset$. Assume node $b \in \mathcal{V}_{G_1}$ is split as in Figure \ref{fig:nonsplit}, transforming $G_1$ into $H_1$.

\begin{figure}[H]
\begin{center}
    \begin{minipage}[t]{0.4\linewidth}
        \centering
        \includegraphics[trim=20mm 25mm 35mm 20mm, clip, width=0.8\textwidth]{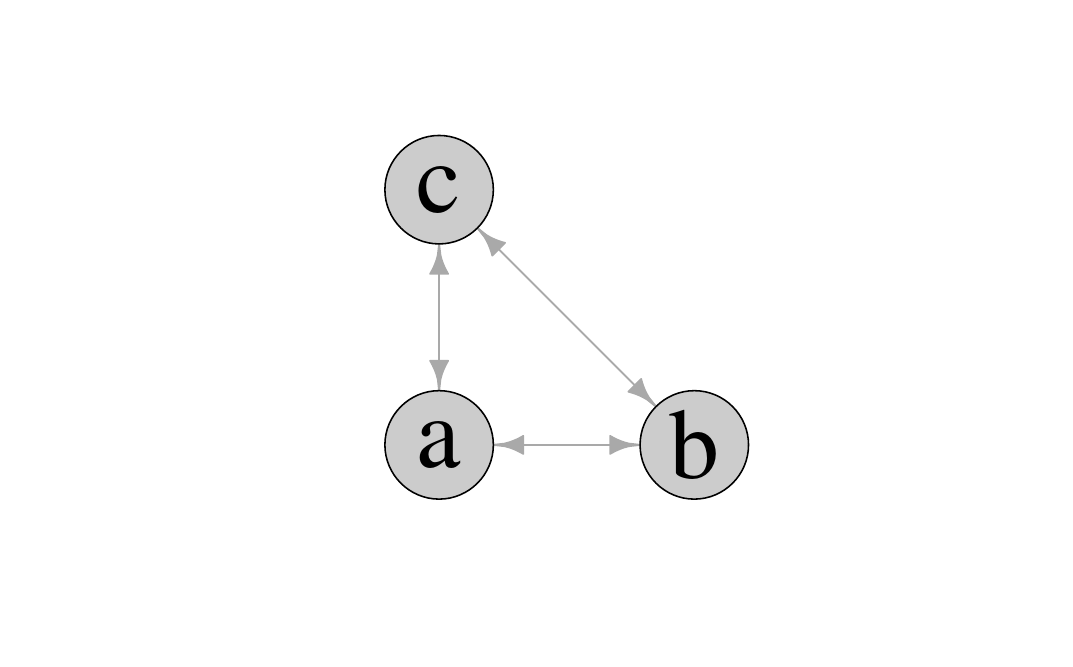}
        \caption*{$G_1$}
    \end{minipage}
    \begin{minipage}[t]{0.4\linewidth}
        \centering
        \includegraphics[trim=20mm 25mm 35mm 20mm, clip, width=0.8\textwidth]{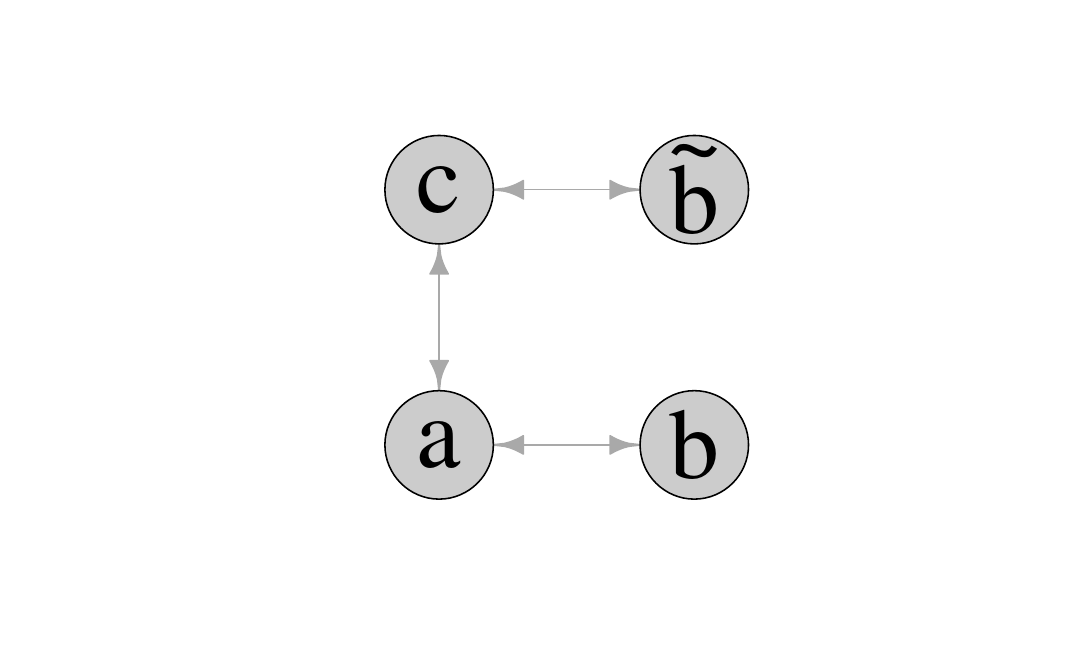}
        \caption*{$H_1$}
    \end{minipage}
\end{center}
\vspace{-5mm}
\caption{Node $b$ is split by adding node $\tilde{b} \notin \mathcal{V}_G$ and replacing edges $\{(b,c), (c,b)\}$ with $\{(\tilde{b}, c), (c, \tilde{b})\}$.}
\label{fig:nonsplit}
\end{figure}

We obtain:
\[
\sum_{v,w \in \mathcal{V}_{H_1}, v \neq w} \frac{1}{l^{H_1}_{vw}} = 2 \cdot \left(1 + 1 + \frac{1}{2} + \frac{1}{2} + \frac{1}{3} + 1\right) = \sum_{v,w \in \mathcal{V}_{G_1}, v \neq w} \frac{1}{l^{G_1}_{vw}} + \frac{8}{3},
\]
and since $G_1$ and $G_2$ are not connected, we have:
\[
\sum_{v,w \in \mathcal{V}_H, v \neq w} \frac{1}{l^H_{vw}} = \sum_{v,w \in \mathcal{V}_{H_1}, v \neq w} \frac{1}{l^{H_1}_{vw}}, \quad \sum_{v,w \in \mathcal{V}_G, v \neq w} \frac{1}{l^G_{vw}} = \sum_{v,w \in \mathcal{V}_{G_1}, v \neq w} \frac{1}{l^{G_1}_{vw}}.
\]
Thus, \eqref{eq:nonsplit} becomes:
\[
\frac{N-1}{N+1} \left( \sum_{v,w \in \mathcal{V}_G, v \neq w} \frac{1}{l^G_{vw}} + \frac{8}{3} \right) \leq \sum_{v,w \in \mathcal{V}_G, v \neq w} \frac{1}{l^G_{vw}}.
\]
This inequality is violated when $N$ is large enough, which can be achieved by increasing the size of $G_2$. Hence, if $h:\mathbb{R} \to \mathbb{R}$ is strictly increasing with $h(0)=0$, condition \textit{C2} is not satisfied.
\end{example}

\subsubsection{Functionality Based on Network Communicability}\label{sec:commun}
Taking the average over all the  communicabilities yields the global measure
\begin{equation}\label{eq:comm}
    \mathcal{F}(G) = \frac{1}{\vert\mathcal{V}_G\vert^2} \sum_{v,w} \mathfrak{c}_{vw}^G
\end{equation}
for a network $G$. 
In contrast to shortest paths, it may be useful to include self-communicabilities, that are entities $\mathfrak{c}_{vv}$, into the consideration since walks of from a node to itself correspond to closed communication loops.  Basing the cost function on communicability we do not only satisfy \textit{C2} in case of edge deletions, but---in contrast to shortest path based costs---also when risk-reduction is achieved by node splitting. 

\begin{proposition}\label{prop:edelcomm}
Consider an acceptance set $\mathcal{A}$ and a set of admissible interventions $\mathcal{I}$ which is risk-reducing for $\mathcal{A}$ such that $\mathcal{I} \subset \mathcal{I}_{e\_ del}\cup \mathcal{I}_{split}$.
 Then $\mathcal{C}$ given by \eqref{eq:cost} and \eqref{eq:comm} and a strictly increasing function $h$ with $h(0) = 0$  is a cost function for $(\mathcal{A},\mathcal{I})$ which satisfies \eqref{eq:C3}. 
\end{proposition}

\begin{proof} \textit{C1} follows from $h(0) = 0$. As in the proof of Proposition~\ref{prop:edelshort} we prove \textit{C2} and \eqref{eq:C3} by showing that $\mathcal{F}(G)$ is non-increasing whenever we apply a intervention in $\mathcal{I}_{e\_ del}\cup \mathcal{I}_{split}$ to $G$, and even strictly decreasing if the intervention alters $G$.

\smallskip\noindent
Let $G\in \mathcal{G}$ with  adjacency matrix $A_G$, given some enumeration of the nodes $\mathcal{V}_G$. Let us first consider edge deletions. To this end, let $\tilde{G} = \kappa^{(v_i, v_j)}_{e\_ del}(G)$ where $v_i,v_j\in \mathcal{V}_G$. If $(v_i, v_j)\in\mathcal{E}_G$, then  $A_{\tilde G}(i,j)<A_G(i,j)$, and indeed $A^k_{\tilde G}(l, m)<A^k_G(l, m)$ for every $k\geq 1$ and $l,m$ such that $(v_i,v_j)$ is element of a walk of length $k$ from $v_l$ to $v_m$. 
Hence, $\mathfrak{c}^{\tilde{G}}_{v w} < \mathfrak{c}^G_{v w}$  for all $v, w\in\mathcal{V}_G$ where $(v_i, v_j)$ lies on a walk from $v$ to $w$, yielding $\mathcal{F}(\tilde{G}) < \mathcal{F}(G)$. Therefore, we conclude that $\mathcal{F}(\tilde{G}) \leq \mathcal{F}(G)$ and $\mathcal{F}(\tilde{G}) = \mathcal{F}(G)$ if and only if $G=\tilde G$. Conversely, if $\tilde G$ is obtained by adding an edge to the network $G$, and thus $G$ is obtained by an edge deletion in $\tilde G$, the previous arguments show that $\mathfrak{c}^{\tilde{G}}_{v w} > \mathfrak{c}^G_{v w}$, and thus $\mathcal{F}(\tilde{G})>\mathcal{F}(G)$ unless $G=\tilde G$. 

Next let us consider node splittings. If $\tilde{G}$ results from a node splitting intervention in $G$, then, conversely, $G$ can be obtained by a node merging intervention in $\tilde{G}$. 
Therefore, it suffices to prove that $\mathcal{F}$ is increasing under node merging.
So let $G=\kappa_{merge}^{v_i,v_j}(\tilde G)$ be the graph obtained after merging two nodes $v_i,v_j\in\mathcal{V}_{\tilde{G}}$ in a network $\tilde{G}$ of size $N$. In case this merge operation is non-trivial, that is $G\neq \tilde G$, assuming $i<j$, this intervention can be described by a matrix operation on the adjacency matrix $A_{\tilde{G}}$, i.e., $A_{G} = \hat{M}\cdot A_{\tilde{G}}\cdot \hat{M}^T$ with  $\hat{M}\in \{0,1\}^{(N-1)\times N}$ defined as
   \begin{equation*}
\hat{M}(l,m) =
        \begin{cases}
               1,& \text{for } l=m<j\\
               1,& \text{for } l=i, m=j\\
               1,& \text{for } j\leq l\leq N-1, m=l+1\\
               0,& \text{else}.
        \end{cases}
    \end{equation*}
One verifies that the product $\hat M^T \hat M$ equals $1$ on the diagonal. Let $I_N\in \mathbb{R}^{N\times N}$ denote the identity matrix. Then it follows that $\hat M^T \hat M\geq I_N$ in the sense that componentwise we have $(\hat M^T \hat M)(l,m)\geq I_N(l,m)$ for all $l,m=1,\ldots,N$.
Since all entries in $A_{\tilde{G}}$ are non-negative we estimate
\begin{align*}
\begin{split}
    \big(A_{G}\big)^k &= \big(\hat{M}\cdot A_{\tilde{G}}\cdot \hat{M}^T\big)^k\\ &= \hat{M}\cdot A_{\tilde{G}}\cdot \hat{M}^T\cdot\hat{M}\cdot A_{\tilde{G}}\cdot\hat{M}^T\cdot\hat{M}\cdot A_{\tilde{G}}\cdots A_{\tilde{G}}\hat{M}^T \\
    &\geq \hat{M}\cdot A_{\tilde{G}}\cdot I_N\cdot A_{\tilde{G}}\cdot I_N\cdot A_{\tilde{G}}\cdots A_{\tilde{G}}\hat{M}^T \\ &= \hat{M}\cdot (A_{\tilde{G}})^k\cdot \hat{M}^T.
    \end{split}
\end{align*}
Thus we obtain
\begin{equation*}
    \sum_{l=1}^{N-1}\sum_{m=1}^{N-1} (A_{G})^k (l,m) \geq \sum_{l=1}^{N-1}\sum_{m=1}^{N-1} (\hat{M}\cdot (A_{\tilde{G}})^k\cdot \hat{M}^T) (l,m) = \sum_{l=1}^{N}\sum_{m=1}^{N} (A_{{\tilde{G}}})^k (l,m)
\end{equation*}
for all $l,m =1,\cdots , N$, and therefore  
\begin{equation}\label{eq:commproof}
    \sum_{v_i,v_j\in\mathcal{V}_G} \mathfrak{c}_{v_i v_j}^G \geq \sum_{v_i,v_j\in\mathcal{V}_{\tilde{G}}} \mathfrak{c}_{v_i v_j}^{\tilde{G}}
\end{equation}
for the sum of all communicabilities. Further, since $\tilde{G}$ is larger than $G$, this yields strict inequality when taking the average value, i.e., $\mathcal{F}(\tilde{G}) < \mathcal{F}(G)$. Therefore, we conclude that $\mathcal{F}(\tilde{G}) \leq \mathcal{F}(G)$ for all $\tilde G$ obtained by performing a node splitting on $G$, and $\mathcal{F}(\tilde{G}) = \mathcal{F}(G)$ if and only if $G=\tilde G$


%
\end{proof}

\begin{remark}
Note that we find ``$>$'' in \eqref{eq:commproof} if and only if there is a matrix entry $(l,m)$ with $A_{\tilde{G}} \hat{M}^T \hat{M} A_{\tilde{G}} (l,m) > A_{\tilde{G}} I_N A_{\tilde{G}} (l,m)$, which means that
\begin{equation*}
    A_{\tilde{G}} \hat{M}^T \hat{M} A_{\tilde{G}} (l,m) = \sum_{p=1}^N\sum_{o=1}^N A_{\tilde{G}} (l,o) \hat{M}^T \hat{M}(o,p) A_{\tilde{G}}(p,m) > \sum_{o=1}^N A_{\tilde{G}}(l,o) A_{\tilde{G}}(o,m).
\end{equation*}
This is the case if and only if $A_{\tilde{G}}(l,i)>0$ and $A_{\tilde{G}}(j, m)>0$, or $A_{\tilde{G}}(l,j)>0$ and $A_{\tilde{G}}(i, m)>0$, which means that $\mathcal{N}^{{\tilde{G}}, in}_{v_i}, \mathcal{N}^{{\tilde{G}}, out}_{v_j}\neq\emptyset$, or  $\mathcal{N}^{{\tilde{G}}, out}_{v_i}, \mathcal{N}^{\tilde{G}, in}_{v_j}\neq\emptyset$. In other words, a node split reduces some communicabilities in a network if and only if the split corresponds to an actual separation of contagion channels.
\end{remark}

\section{Examples} \label{sec:examples}
The examples presented in this section will all be based on network acceptance sets $\mathcal{A}$ of the form 
\begin{equation}\label{eq:accset}
  \mathcal{A} = \{G\in \mathcal{G} \mid Q(G) \leq l_{|\mathcal{V}_G|}\}, 
\end{equation}
where a specific network quantity $Q:\mathcal{G}\to\mathbb{R}\cup \{-\infty, \infty\}$ needs to be bounded for acceptability of  a network $G$, and the bounds $l_N\in \mathbb{R}$, $N\in\mathbb{N}$, may depend on the network size $N$. 
\begin{definition}
Let $\mathcal{I}$ be a non-empty set of network interventions and let $Q:\mathcal{G}\to\mathbb{R}\cup \{-\infty, \infty\}$. We call $Q$ $\mathcal{I}$-monotone if  $Q(\kappa(G))\leq Q(G)$ for all $G\in\mathcal{G}$ and $\kappa\in \mathcal{I}$.
\end{definition}

The following lemma is easily verified. 

\begin{lemma}
Suppose that $Q$ is $\mathcal{I}$-monotone. Then
\begin{enumerate}
\item $Q(\kappa(G))\leq Q(G)$ for all $G\in\mathcal{G}$ and $\kappa \in [\mathcal{I}]$, 
 \item $\mathcal{I}$ is risk-reducing for $\mathcal{A}$.
 \end{enumerate}
 If, moreover, $\mathcal{I}$ is not partially self-revers, then $\mathcal{I}$-monotonicity is the same as $Q$ being decreasing with respect to the partial order $\preccurlyeq_{\mathcal{I}}$ given in Lemma~\ref{lem:partial:order}.
\end{lemma}

\subsection{DMFNR Based on Moments of the Degree Distribution }\label{sec:degdistrcont}
In this section we consider examples of the function $Q$ in \eqref{eq:accset} depending on moments of the degree distribution. To this end, we write $k_v^{G, in}= \vert\mathcal{N}^{G, in}_v\vert$ and $k_v^{G, out}= \vert\mathcal{N}^{G, out}_v\vert$ for the in- and out-degrees of nodes $v\in \mathcal{V}_G$, see also Section~\ref{sec:neighbor}. 
Given a network $G\in\mathcal{G}$,  we let $P^{out}_G (k)$ denote the fraction of nodes $v$ with out-degree $k$, or, equivalently, the probability that a node $v$ which is chosen uniformly at random comes with out-degree $k$. The corresponding  distribution $P^{out}_G$ over all possible out-degrees is called the \textit{out-degree distribution} of the network $G$. Analogously, we can define the \textit{in-degree distribution} $P^{in}_G$. Let $K^{out}_G$ denote a random variable which represents the out-degree of a randomly chosen node given the probability distribution $P^{out}_G$, i.e.\ $\mathbb{P}(K^{out}_G=k)=P^{out}_G (k)$ for all $k=0,1,\ldots, \vert\mathcal{V}_G\vert -1$. 
We define $K^{in}_G$ analogously.
The \textit{n-th moment} of the in- and out-degree distributions are given by 
\begin{equation*}
\mathbb{E}\big[(K^{in}_G)^n\big] := \sum_{k\in\mathbb{N}_0} k^n P^{in}_G(k),\quad \mathbb{E}\big[(K^{out}_G)^n\big] := \sum_{k\in\mathbb{N}_0} k^n P^{out}_G(k),\quad n\in\mathbb{N}.
\end{equation*}

Clearly any network acceptance set such that $Q$ only depends on the degree distribution is topological invariant:

\begin{lemma}\label{lem:linv}
If $\mathcal{A}$ is given by  \eqref{eq:accset} where $Q$ only depends on the in-, out-degree distribution, 
then $\mathcal{A}$ satisfies P4. 
\end{lemma}

\begin{remark}\label{rem:incoming}
As regards in- and out-degree distributions, note that the following examples indicate that for the management of systemic cyber risk, the distribution of outgoing, not incoming, node degrees is the relevant entity. 
\end{remark}

\subsubsection{Variance of Degree Distributions}
Obviously, not all functions $Q$ based on moments of node degrees are suited risk controls. Indeed, in the following we show that the variance of either in- or out-degree distributions  is not a satisfactory control. 
\begin{proposition}
Suppose that the network acceptance set $\mathcal{A}$ satisfies \eqref{eq:accset} where $Q$ equals the variance of either the in- or out-degree distribution, and that $l_N\geq 0$ for all $N\in \mathbb{N}$. In both cases, $\mathcal{A}$ satisfies P1, but violates P3. 
Moreover, $Q$ is not $\mathcal{I}$-monotone for any choice of $\mathcal{I}\subset \mathcal{I}_{e\_ del}\cup \mathcal{I}_{e\_add} \cup \mathcal{I}_{n\_del} \cup \mathcal{I}_{n\_add} $ such that $\mathcal{I}\neq \emptyset$, that is for no non-trivial set of elementary network interventions. 
\end{proposition}
\begin{proof}  
The in- and out-degree distributions of all edgeless 
 and all complete graphs have zero variance. Thus, P1 is always satisfied while P3 is violated. 
    
 As regards the monotonicity of $Q$,  variance, of either the in- our out-degree distribution, increases by adding edges to the edgeless graph, or by deleting edges in the complete graph.  Moreover,  the variance of a complete graph  increases when adding an isolated node to the network. As regards node deletion, note that 
 any circular graph given by $\mathcal{V}_{G^\circ} = \{v_1,\cdots , v_N\}$ and $$\mathcal{E}_{G^\circ} = \{ (v_1, v_2), (v_2, v_1),\cdots, (v_{N-1}, v_N), (v_{N}, v_{N-1}), (v_N, v_1), (v_1, v_N)\},$$ see Figure \ref{fig:ring}, comes with zero variance for both in- and out-degree distributions since each node has the in- and out-degree 2. However, when deleting node $v_N$, the network equals the undirected line graph, see Figure \ref{fig:ring}, with vertex set $\{v_1, \cdots, v_{N-1}\}$ and edges $$\{(v_1, v_2), (v_2, v_1),\cdots, (v_{N-2}, v_{N-1}), (v_{N-1}, v_{N-2})\}.$$ If $N\geq 4$, in contrast to all the other nodes, $v_1$  and $v_{N-1}$ only have one incoming and one outgoing edge. Therefore, the variance of both the in- and out-degree distribution now is positive. 
\end{proof}
\begin{figure}[h]
\begin{center}

		\begin{minipage}[t]{0.32\linewidth}
			\centering
			{\includegraphics[trim={4cm 2.5cm 4cm 2.2cm},clip,width=1\textwidth]{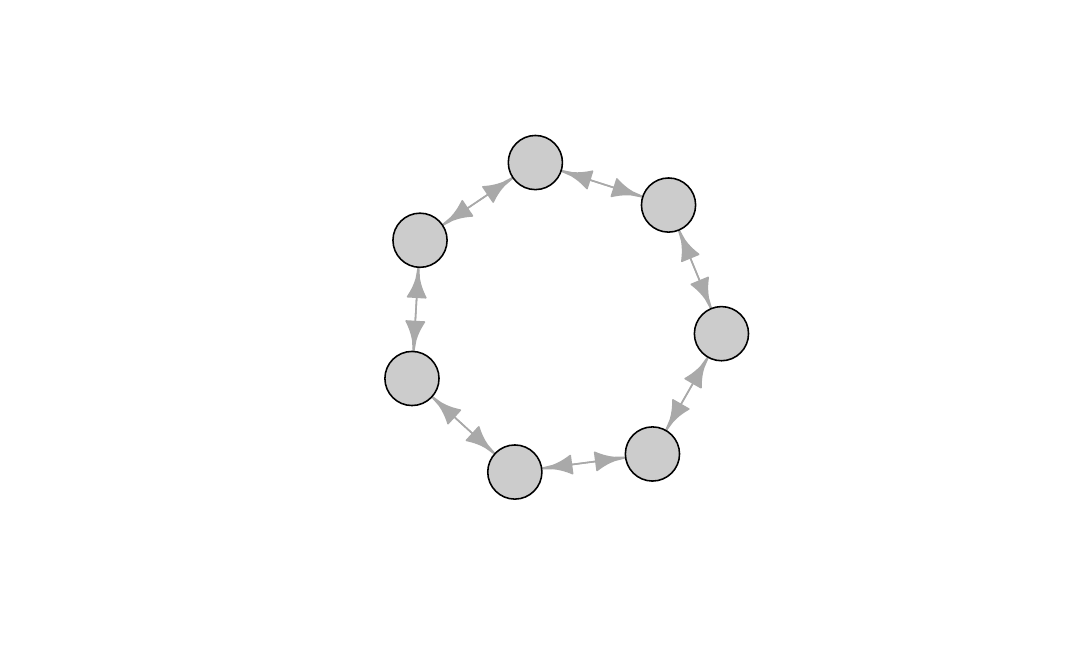}}
		\end{minipage}
        \begin{minipage}[t]{0.32\linewidth}
			\centering
                  {\includegraphics[trim={4cm 2.5cm 4cm 2.2cm},clip,width=1\textwidth]{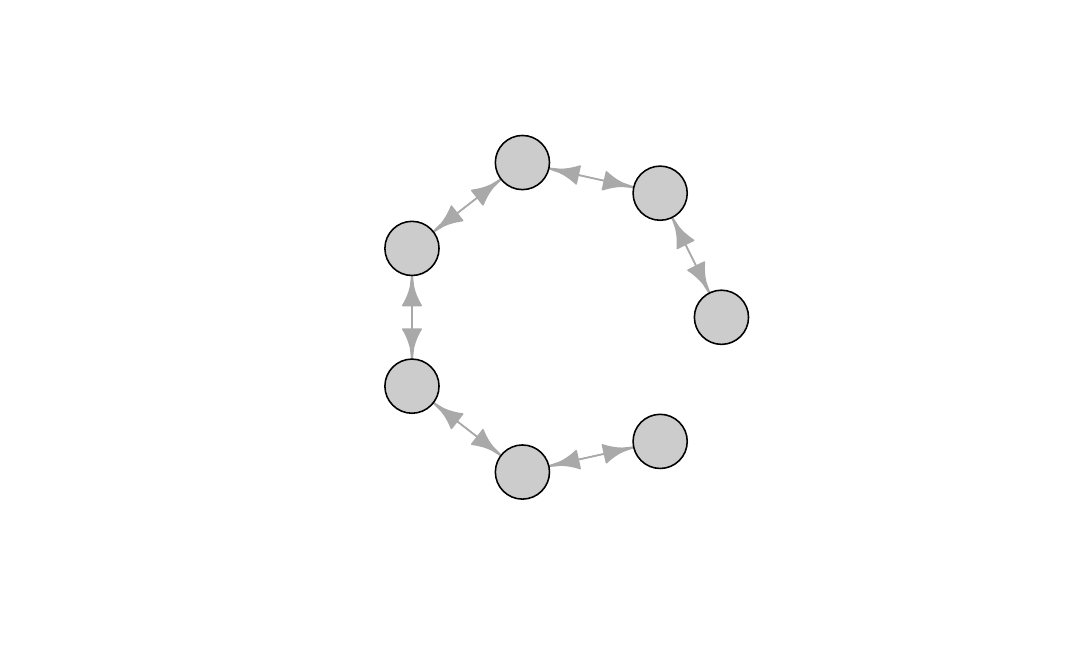}}	
        \end{minipage}
        \begin{minipage}[t]{0.32\linewidth}
			\centering
                  {\includegraphics[trim={4cm 2.5cm 4cm 2.2cm},clip,width=1\textwidth]{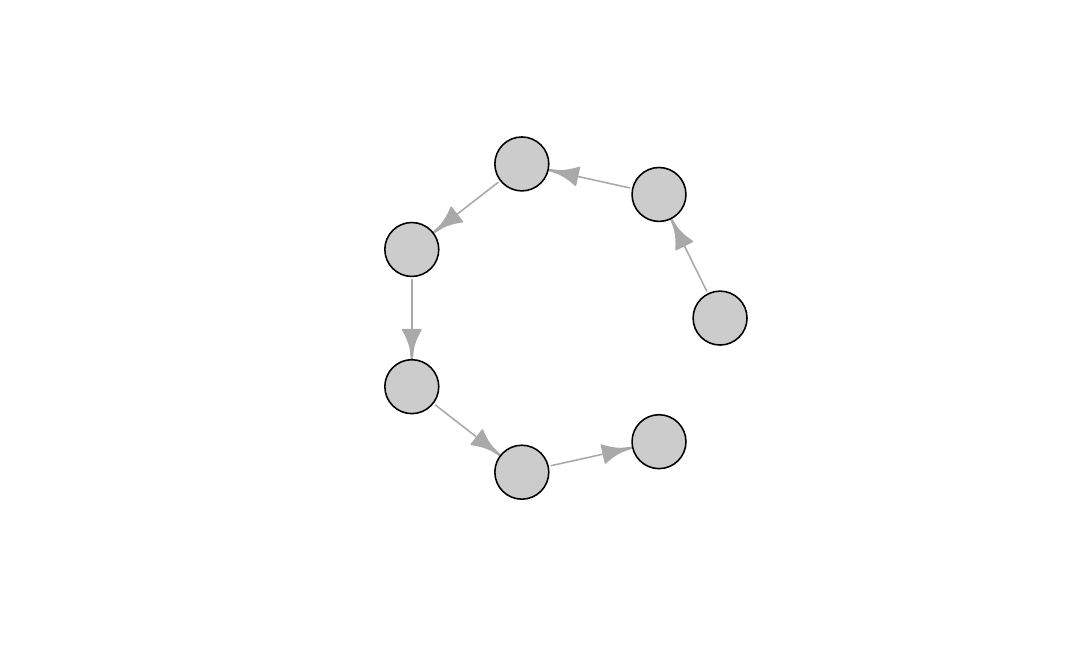}}	
        \end{minipage}
		\end{center}
		\caption{An undirected ring graph (left), an undirected line graph (middle), and a directed line graph (right), all consisting of $N=7$ nodes.}
		\label{fig:ring}
\end{figure}

\subsubsection{Average Degrees}

Next we consider the  \textit{average in-} and \textit{-out degree} of a network $G$. Note that these first moments are equal since every edge that emanates from a node  $v$ needs to arrive at another one, that is we have
\begin{equation}\label{eq:avdeg}
    \mathbb{E}\big[K^{in}_G\big] = \frac{1}{N} \sum_{i=1}^N k_{v_i}^{G, in} = \frac{\vert\mathcal{E}_G\vert}{\vert\mathcal{V}_G\vert} = \frac{1}{N} \sum_{i=1}^N k_{v_i}^{G, out} = \mathbb{E}\big[K^{out}_G\big].
\end{equation}

\begin{proposition}\label{prop:average}
 Consider $\mathcal{A}$ as in \eqref{eq:accset} with $Q(G) =\mathbb{E}\big[K^{out}_G\big] (=\mathbb{E}\big[K^{in}_G\big])$. Then $Q$ is $\mathcal{I}$-monotone for any $\mathcal{I}$ that is composed of edge deletions and node splittings, i.e.\ $\mathcal{I}\subset \mathcal{I}_{e\_del}\cup \mathcal{I}_{split}$. Hence, $\mathcal{I}$ is risk-reducing for  $\mathcal{A}$ whenever $\mathcal{I}\subset \mathcal{I}_{e\_del}\cup \mathcal{I}_{split}$. 
\end{proposition}
\begin{proof}
The result follows directly from Equation \eqref{eq:avdeg}: For any $\kappa\in\mathcal{I}_{e\_del}$ and network $G$ we have that $\vert\mathcal{V}_{\kappa(G)}\vert = \vert\mathcal{V}_G\vert$ and $\vert\mathcal{E}_{\kappa(G)}\vert\leq \vert\mathcal{E}_G\vert$. In case of node splittings $\kappa\in\mathcal{I}_{split}$ we have $\vert\mathcal{E}_{\kappa(G)}\vert = \vert\mathcal{E}_G\vert$ and $\vert\mathcal{V}_{\kappa(G)}\vert\geq \vert\mathcal{V}_G\vert$.
\end{proof}

\begin{corollary}
Suppose that $\mathcal{I}$ is non-empty and satisfies $ \mathcal{I} \subset \mathcal{I}_{e\_del}\cup \mathcal{I}_{split}$, and suppose that $l_{N_0}\geq 0$ for some $N_0\in \mathbb{N}$. Then, $\mathcal{A}$ given in \eqref{eq:accset} with $Q(G) =\mathbb{E}\big[K^{out}_G\big]$ is a topology invariant network acceptance set. 
$(\mathcal{A}, \mathcal{I}, \mathcal{C})$ is a DMFNR for any cost function $\mathcal{C}$ for $(\mathcal{A}, \mathcal{I})$.
\end{corollary}

\begin{proof} Combine Proposition~\ref{prop:average} with Lemma~\ref{lem:linv}. Note that $l_{N_0}\geq 0$ for some $N_0\in \mathbb{N}$ ensures that $\mathcal{A}$ is non-empty.
\end{proof}

\begin{remark}
Propositions~\ref{prop:edelshort} and \ref{prop:edelcomm} show how, apart from monetary costs, the cost function $\mathcal{C}$ in Proposition~\ref{prop:average} can be based on a loss of network functionality. Communicability-based costs work for any choice of $\mathcal{I}$ as in Proposition~\ref{prop:average}, whereas costs based 
 on shortest paths require a further restriction of  $\mathcal{I}$ to be a subset of $\mathcal{I}_{e\_del}$. 
\end{remark}

\begin{proposition}\label{prop:firstmom}
The acceptance set $\mathcal{A}$ as in \eqref{eq:accset} with $Q(G) =\mathbb{E}\big[K^{out}_G\big] (=\mathbb{E}\big[K^{in}_G\big])$ cannot satisfy P2 and  P3 simultaneously. More precisely, if P2 is satisfied, then P3 is violated, and vice versa. 
\end{proposition}

For the following and later proofs the notion of a \textit{directed star graph} will be important: ${G}^\ast=(\mathcal{V}_{{G}^\ast},\mathcal{E}_{{G}^\ast})$ is a directed star graph with super-spreader $v^\ast$ if $\emptyset\neq \mathcal{V}_{{G}^\ast}\subset \mathbb{V}$ and  \[\mathcal{E}_{{G}^\ast} = \{ (v^\ast,w)\mid w\in \mathcal{V}_{{G}^\ast}\setminus \{v^\ast\}\}.\]

\begin{proof}    Let $G$ be a weakly connected network of size $N$. Consider the undirected version of $G$ given by $\tilde{G} = (\mathcal{V}_G, \mathcal{E}_{\tilde{G}})$ with $\mathcal{E}_{\tilde{G}} = \{ (v,w)\vert (v,w)\in\mathcal{E}_G \text{ or } (w,v)\in\mathcal{E}_G\}$. By definition, we have $\vert\mathcal{E}_G\vert\geq\vert\mathcal{E}_{\tilde{G}}\vert /2$, and  $\tilde{G}$ is strongly connected. The latter implies that $\vert\mathcal{E}_{\tilde{G}}\vert\geq 2(N-1)$ since $\tilde{G}$ must contain an undirected spanning tree, and every undirected tree graph of size $N$ contains $2(N-1)$ directed edges, see \cite[Corollary 7 and 8]{Bollobas1998}. This yields $\vert\mathcal{E}_G\vert\geq N-1$, and thus $Q(G)\geq (N-1)/N$. 
 Now consider a directed star graph $G^\ast$ of size $N$. $G^\ast$  only contains $\vert \mathcal{V}_{G^\ast}\setminus\{v^\ast\}\vert = N-1$ edges, thus $Q(G^\ast)=(N-1)/N$. Therefore, if a weakly connected network is acceptable, then the directed star graph of the same size is also acceptable. The latter constitutes a violation of  P3. Conversely, if  P3 is satisfied, so in particular no directed star graph is acceptable, then we cannot have any acceptable weakly connected graph.  
\end{proof}

\begin{remark}
    The average node degree only provides a measure of overall graph connectivity in the mean, considering the number of edges in comparison to the number of nodes. However, it does not evaluate characteristics which relate closer to the topological arrangement of edges within the network, that is for example the presence hubs which may act as risk amplifiers. Networks can have a very different risk profile, even if they contain the same number of edges and nodes,  see for example the analysis in  \cite{Awiszus2023b}. Hence, one may doubt whether acceptance sets given by $Q(G) =\mathbb{E}\big[K^{out}_G\big] (=\mathbb{E}\big[K^{in}_G\big])$ are reasonable choices. In that respect, Proposition~\ref{prop:firstmom} shows that if $Q(G) =\mathbb{E}\big[K^{out}_G\big]$, $\mathcal{A}$ cannot satisfy P2 and P3  simultaneously.  
\end{remark}

\subsubsection{Second  In- and Out-Moments}
\begin{proposition}\label{prop:secmomint}
Consider $\mathcal{A}$ as in \eqref{eq:accset} with $Q(G)= \mathbb{E}\big[ (K_G^{out})^2\big]$ or  $Q(G)=\mathbb{E}\big[ (K_G^{in})^2\big]$. Then $Q$ is $\mathcal{I}$-monotone for any $\mathcal{I}\subset \mathcal{I}_{e\_del}\cup \mathcal{I}_{split}$. 
\end{proposition}
\begin{proof}
  Clearly, the second moments are reduced when deleting edges in a network. 
Further, the second moment of the in- or out-degree distribution is also decreasing under node splitting: Let $v$ be the node in a network $G$ of size $N$ which is split into $v$ and $\tilde{v}$, and $\tilde{G}$ the resulting network. Due to $k_v^{G, in} = k_v^{\tilde{G}, in} + k_{\tilde{v}}^{\tilde{G}, in }$ and $k_v^{G, out} = k_v^{\tilde{G}, out} + k_{\tilde{v}}^{\tilde{G}, out}$, we have 
    \begin{equation}\label{eq:ndeg}
    (k_v^{\tilde{G}, in })^2 + (k_{\tilde{v}}^{\tilde{G}, in })^2\leq (k_v^{\tilde{G}, in})^2 + (k_{\tilde{v}}^{\tilde{G}, in})^2 + 2 k_v^{\tilde{G}, in} k_{\tilde{v}}^{\tilde{G}, in} = (k_v^{\tilde{G}, in} + k_{\tilde{v}}^{\tilde{G}, in})^2 = (k_v^{G,in })^2,
    \end{equation}
    and thus
    \begin{align*}
        \mathbb{E}[(K_{\tilde{G}}^{in})^2]&= \frac{1}{N+1} \Big( \sum_{w\neq v, \tilde{v}} (k_{w}^{\tilde{G}, in})^2 + (k_{v}^{\tilde{G}, in})^2 + (k_{\tilde{v}}^{\tilde{G}, in})^2\Big) \\
        &\leq \frac{1}{N} \left(\sum_{w\neq v, \tilde{v}} (k_{w}^{\tilde{G},in})^2 + (k_v^{G,in})^2\right) =  \frac{1}{N} \left(\sum_{w\neq v, \tilde{v}} (k_{w}^{G,in })^2 + (k_v^{G,in})^2 \right)= \mathbb{E}[(K_{G}^{in})^2] .
    \end{align*}
\end{proof}

\begin{corollary}
Let $\mathcal{A}$ be given by \eqref{eq:accset} with $Q(G)= \mathbb{E}\big[ (K_G^{out})^2\big]$ or  $Q(G)=\mathbb{E}\big[ (K_G^{in})^2\big]$ and let $\mathcal{I}\subset \mathcal{I}_{e\_del}\cup \mathcal{I}_{split}$ be non-empty. Then $\mathcal{I}$ is risk-reducing for $\mathcal{A}$. Thus, $(\mathcal{A},\mathcal{I}, \mathcal{C})$ is DMFNR for any cost function $\mathcal{C}$ for $(\mathcal{A},\mathcal{I})$. 
\end{corollary}

\begin{proposition}\label{prop:secout}
Consider a network acceptance set $\mathcal{A}$ as in \eqref{eq:accset} with $Q(G) = \mathbb{E}\big[ (K_G^{out})^2\big]$ or $Q(G) = \mathbb{E}\big[ (K_G^{in})^2\big]$. 
\begin{enumerate}
\item $\mathcal{A}$ satisfies P1 if and only if $l_N\geq 0$ for all $N\geq 2$,
    \item $\mathcal{A}$ satisfies P2 if and only if there is a $N_0\in \mathbb{N}$ such that $l_N\geq (N-1)/N$ for all $N\geq N_0$,
    \item If $Q(G) = \mathbb{E}\big[ (K_G^{out})^2\big]$, then $\mathcal{A}$ satisfies P3 if and only if $l_N < (N-1)^2/N$ for all $N\geq 1$.
    \item If $Q(G) = \mathbb{E}\big[ (K_G^{in})^2\big]$, then $\mathcal{A}$ cannot satisfy the P2 and P3 simultaneously. 
\end{enumerate}
\end{proposition}
\begin{proof}
1. The second moment of the (out- or ingoing) node degrees of any edgeless graph equals zero by definition.

\smallskip\noindent
2.  We only prove the assertion for the out-degree distribution, The same proof holds true in case of the in-degree distribution.  Since  $k\leq k^2$ for all $k\in\mathbb{N}$, we find that
 \begin{equation}\label{eq:second:first:moment}
      \mathbb{E}\big[ (K_G^{out})^2\big] \geq \frac{1}{N} \sum_{v\in\mathcal{V}_G} k_v^{G, out} = \frac{\vert \mathcal{E}_G\vert}{N} = \mathbb{E}\big[ K_G^{out}\big]
 \end{equation}
 for any network $G$ of size $N$. Recall that a  weakly connected network $G$ of size $N$ contains at least  $N-1$  edges, see the proof of Proposition \ref{prop:firstmom}.  Therefore, any  weakly connected network $G$ satisfies 
 \begin{equation}\label{eq:inineq}
      \mathbb{E}\big[ (K_G^{out})^2\big] \geq \frac{N-1}{N}. \end{equation}
 Equality  in \eqref{eq:inineq} is realized by every directed line graph $G^L$ where $\mathcal{V}_{G^L} = \{v_1,\cdots , v_N\}$ and $\mathcal{E}_{G^L} = \{ (v_1, v_2), (v_2, v_3), \cdots, (v_{N-1}, v_N)\}$, see Figure \ref{fig:ring}.  Therefore, $\mathcal{A}$ contains a weakly connected graph if and only if  $\mathcal{A}$ contains a directed line graph of the same size. 

\smallskip\noindent
 3.  According to Proposition \ref{prop:secmomint}, $Q$  is  $\mathcal{I}_{e\_del}$-monotone. Therefore, it suffices to consider the networks which contain a super-spreader but otherwise have minimal amount of edges, that is the directed star graphs. The second moment of the out-degree distribution of a directed star graph $G^\ast$ of size $N$ is given by 
 \begin{equation*}
     \mathbb{E}\big[ (K^{out}_{G^\ast})^2\big] = \frac{1}{N} \big( 1\cdot (N-1)^2 + (N-1)\cdot 0\big) = \frac{(N-1)^2}{N}.
 \end{equation*}
 
 \smallskip\noindent
4.  Similar to \eqref{eq:second:first:moment} above we obtain that 
 \begin{equation}\label{eq:secmominc}
    \mathbb{E}[(K_G^{in})^2] \geq  \frac{\vert \mathcal{E}_G\vert}{N} 
 \end{equation}
 for any network $G$ of size $N$. 
 Consider any  weakly connected network $G$ with $N$ nodes for some $N\in \mathbb{N}$.
 Then \eqref{eq:secmominc} yields $\mathbb{E}[(K_G^{in})^2] \geq (N-1)/N$ (proof of Proposition \ref{prop:firstmom}). A directed star graph $G^\ast$ of size $N$  satisfies $$\mathbb{E}[(K_{G^\ast}^{in})^2] =\frac{1}{N}((N-1)\cdot 1 + 1\cdot 0) = \frac{N-1}{N},$$ because the super-spreader has in-degree 0, and all the other nodes come with a in-degree of 1.  Hence, acceptability of a weakly connected graph of size $N$ implies $l_N\geq (N-1)/N$, and thus that the directed star graph of the same size is acceptable, which violates P3. Conversely, if we exclude any directed star graph, then necessarily  $l_N< (N-1)/N$ for all $N\in \mathbb{N}$, so there cannot be an acceptable weakly connected graph, so P2 is violated.

\end{proof}

\subsection{DMFNR Based on Hub Control via Centrality Measures}\label{sec:hub}
In this section we consider as control $Q$ in \eqref{eq:accset} the maximal centrality 
\begin{equation}\label{eq:hubcontrol}
Q(G) = \max_{v\in\mathcal{V}_G}\mathfrak{C}(v, G),
\end{equation}
where the \textit{centrality} $\mathfrak{C}(v, G)$ of a node $v$ of graph $G$ will be given by a centrality measure $\mathfrak{C}$ as discussed in the following Section~\ref{sec:ncentr}. Controlling the maximal centrality limits the impact of the most central/systemic nodes within the network. 


\subsubsection{ Node Centralities} \label{sec:ncentr}
 A {\em node centrality measure} is a map
    \begin{equation*}
    \mathfrak{C}^n:\mathbb{V}\times\mathcal{G}\to \mathbb{R}_+\cup\{\emptyset\} \qquad \mbox{such that} \quad \mathfrak{C}(v, G) = \emptyset\Leftrightarrow v\notin \mathcal{V}_G.
\end{equation*}
The simplest way to measure { the centrality of a node is  by its} in-, out-degree:
      \begin{equation*}
       \mathfrak{C}^{deg}_{in}(v, G) = k^{G, in}_v,\quad \mathfrak{C}^{deg}_{out}(v, G) = k^{G, out}_v    
        \end{equation*}
in case $v\in \mathcal{V}_G$, and $\mathfrak{C}^{deg}_{in}(v, G)=\mathfrak{C}^{deg}_{out}(v, G)=\emptyset$ if $v\notin \mathcal{V}_G$. 
        
There are a number of  extensions of the concept of degree centrality where connections to high-degree nodes are more important than those to nodes with a low degree level. Usually, this leads to a definition of centrality that is based on the entries of the normalized (left or right) eigenvector associated with the largest eigenvalue of the adjacency matrix, see Section 7.1 in \cite{Newman2018} for details. 

Alternatively, centrality can also be defined in terms of shortest paths instead of node degrees. 
\textit{In-} or \textit{out-closeness centrality}, respectively, calculate the average distance from a node to others, either in terms of incoming or outgoing paths:
\begin{equation}\label{eq:closecentr}
\mathfrak{C}_{in}^{close} (v, G) = \frac{1}{|\mathcal{V}_G|-1} \sum_{w\neq v} \frac{1}{l^G_{wv}},\quad \mathfrak{C}_{out}^{close} (v, G) = \frac{1}{|\mathcal{V}_G|-1} \sum_{w\neq v} \frac{1}{l^G_{vw}}
\end{equation}
in case $v\in\mathcal{V}_G $ and $|\mathcal{V}_G|\geq 2$,  $\mathfrak{C}_{in}^{close} (v, G)=\mathfrak{C}_{out}^{close} (v, G)=0$ whenever $G=(\{v\}, \emptyset)$, and $\mathfrak{C}_{in}^{close} (v, G)=\mathfrak{C}_{out}^{close} (v, G)=\emptyset$ if $v\notin \mathcal{V}_G$. Note that we set $1/\infty:=0$.

Another prominent example of a path-based centrality measure in the literature is {\em betweenness centrality} 
    \begin{equation}   \label{eq:betnode}    
        \mathfrak{C}^\text{bet}(v, G) = \sum_{\substack{u, w \\ u,w\neq v}} \frac{\sigma^G_{uw}(v)}{\sigma^G_{uw}} \quad \mbox{if}\,  v\in \mathcal{V}_G, \quad \mathfrak{C}^\text{bet}(v, G) =\emptyset \quad \mbox{if}\,  v\notin \mathcal{V}_G,
		\end{equation}          
         where $\sigma^G_{uw}$ denotes the total number of shortest paths from node $u$ to $w$ in $G$, and $\sigma^G_{uw}(v)$ is the number of these paths that go through node $v$, and where we set $0/0:=0$. Note that slightly different definitions can be found in the literature, f.e., in \cite{Newman2018}, where also paths with $u = v$ or $w=v$ are considered.

 \subsubsection{DMFRN based on In- and Out-Degree Centrality}
\begin{proposition}\label{prop:degmon}
    Consider an acceptance set $\mathcal{A}$ as in \eqref{eq:accset} with $Q(G) = \max_{v\in\mathcal{V}_G}\mathfrak{C}^{deg}_{out}(v, G)$ or $Q(G) = \max_{v\in\mathcal{V}_G}\mathfrak{C}^{deg}_{in}(v, G)$. Then $Q$ is $\mathcal{I}$-monotone for any non-empty $\mathcal{I}\subset \mathcal{I}_{e\_del}\cup \mathcal{I}_{split}$, and $(\mathcal{A},\mathcal{I},\mathcal{C})$ is a DMFNR whenever $\mathcal{C}$ is a cost function for $(\mathcal{A},\mathcal{I})$.
\end{proposition}
\begin{proof}
    Clearly, deleting an edge does not increase the in- or out-degree of any node in a graph $G$. The same holds when an arbitrary node split is applied.
\end{proof}

\begin{proposition}\label{prop:outindegree}
Consider an acceptance set $\mathcal{A}$ as in \eqref{eq:accset} with $Q(G) = \max_{v\in\mathcal{V}_G}\mathfrak{C}^{deg}_{out}(v, G)$ or $Q(G) = \max_{v\in\mathcal{V}_G}\mathfrak{C}^{deg}_{in}(v, G)$. Then 
\begin{enumerate}
\item $\mathcal{A}$ satisfies P1 if and only if $l_N\geq 0$ for all $N\geq 2$,
 \item $\mathcal{A}$ satisfies  P2 if and only if there is a $N_0$ such that $l_N\geq 1$ for all $N\geq N_0$.
 \item If $Q(G) = \max_{v\in\mathcal{V}_G}\mathfrak{C}^{deg}_{out}(v, G)$, then $\mathcal{A}$ satisfies P3 if and only if $l_N < N-1$ for all $N\geq 1$.
 \item If $Q(G) = \max_{v\in\mathcal{V}_G}\mathfrak{C}^{deg}_{in}(v, G)$, then $\mathcal{A}$ cannot satisfy P2  and P3 simultaneously. 
 \end{enumerate}
\end{proposition}
\begin{proof}
    \begin{enumerate}
        \item The in- and out-degree of every node in an edgeless graph equals zero. 
        \item If a network $G\in\mathcal{G}$ with $\vert\mathcal{V}_G\vert\geq 2$ is  weakly connected, then  there are nodes $v, w\in\mathcal{V}_G$ with $\vert\mathcal{N}_v^{G, out}\vert, \vert\mathcal{N}_w^{G, in}\vert \geq 1$. 
 This proves necessity. Sufficiency follow by considering the directed line graph of size $N$ (see Figure~\ref{fig:ring}). 
        \item The out-degree of a super-spreader equals $N-1$ by definition. 
        \item $Q$ is $\mathcal{I}_{e\_ del}$-monotone according to Proposition \ref{prop:degmon}. Therefore, if $\mathcal{A}$ would contain a network with a super-spreader, then $\mathcal{A}$  would in particular also contain the directed star graph of the same size. For a directed star graph $G^\ast$, however,  of size $\vert\mathcal{V}_{G^\ast}\vert\geq 2$ with super-spreader $v^\ast$, we find $\mathfrak{C}^{deg}_{in}(v, G)=1$ for all $v\in\mathcal{V}_G\setminus\{v^\ast\}$, and $\mathfrak{C}^{deg}_{in}(v^\ast, G) = 0$. Thus, $Q(G^\ast) = 1$. Recall that for any weakly connected network $G$  with $\vert\mathcal{V}_G\vert\geq 2$ we have $Q(G) = \max_{v\in\mathcal{V}_G}\mathfrak{C}^{deg}_{in}(v, G)\geq 1$, see 2. Hence, 4. follows. 
    \end{enumerate}
\end{proof}

\subsubsection{ In- and Out-Closeness Centrality}

We first note that for path-based centrality measures node splits may indeed worsen the situation. 

\begin{example}
    Consider again the setting from Example \ref{ex:nonsplit}. For all nodes $v = a, b, c$ from the initial network component $G_1$, we have
\begin{equation*}
    \mathfrak{C}^{close}_{in}(v, G) = \mathfrak{C}^{close}_{out}(v, G) = \frac{1}{N-1} \Big(\frac{1}{1} + \frac{1}{1}\Big) = \frac{2}{N-1}.
\end{equation*}
Now, after the node split, the nodes $a$ and $c$ come with  in- and out-closeness centralities of
\begin{equation*}
    \mathfrak{C}^{close}_{in} (v, H) = \mathfrak{C}^{close}_{out} (v, H) = \frac{1}{N} \Big( \frac{1}{1} + \frac{1}{1} + \frac{1}{2}\Big) = \frac{5}{2N},\quad v= a,c,
\end{equation*}
and we find that
\begin{equation*}
    \max_{v\in\mathcal{V}_G}\mathfrak{C}^{close}_{\ast}(v, G) < \max_{v\in\mathcal{V}_H}\mathfrak{C}^{close}_{\ast}(v, H) \Leftrightarrow \frac{2}{N-1} < \frac{5}{2N} \Leftrightarrow 5 < N,\quad v= a,c,
\end{equation*}
for $ \ast =  in, out$. Thus, the maximal in- and out-closeness centrality is increased under the node split if the component $G_2$ consists of at least three isolated nodes.

Similarly, the betweenness centrality as defined in \eqref{eq:betnode} of all nodes $v\in\mathcal{V}_{G_1}$ equals zero. However, after the split of node $b$, we have that the shortest paths from node $b$ to $\tilde{b}$ and vice versa both pass through nodes $a$ and $c$, thus
\begin{equation*}
    Q(H) = \mathfrak{C}^{bet}(a, H_1) = \mathfrak{C}^{bet}(c, H_1) = 2 > 0 = Q(G).
\end{equation*}
\end{example}

\begin{proposition}\label{prop:closeness}
Consider a set $\mathcal{A}\subset\mathcal{G}$ as in \eqref{eq:accset} with $Q(G) = \max_{v\in\mathcal{V}_G}\mathfrak{C}^{close}_{out}(v, G)$ or $Q(G) = \max_{v\in\mathcal{V}_G}\mathfrak{C}^{close}_{in}(v, G)$. Then 
\begin{enumerate}
\item $Q$ is $\mathcal{I}$-monotone for any non-empty $\mathcal{I}\subset \mathcal{I}_{e\_del}\cup \mathcal{I}_{s\_iso}\cup\{\operatorname{id}\}$. In that case  $(\mathcal{A}, \mathcal{I}, \mathcal{C})$ is a DMFNR whenever $\mathcal{C}$ is a cost function $(\mathcal{A}, \mathcal{I})$. 
\item $\mathcal{A}$ satisfies P1 if and only if $l_N\geq 0$ for all $N\geq 2$,
 \item $\mathcal{A}$ satisfies  P2 if and only if there is a $N_0\in \mathbb{N}$ such that $l_N\geq 1/(N-1)$ for all $N\geq N_0$.
\item If $Q(G) = \max_{v\in\mathcal{V}_G}\mathfrak{C}^{close}_{out}(v, G)$, then $\mathcal{A}$ satisfies  P3 if and only if $l_1<0$ and $l_N < 1$ for all $N\geq 2$.
 \item If $Q(G) = \max_{v\in\mathcal{V}_G}\mathfrak{C}^{close}_{in}(v, G)$, then $\mathcal{A}$ cannot satisfy P2 and P3 simultaneously. More precisely, P2 implies that P3 is violated and vice versa.
\end{enumerate}
\end{proposition}
\begin{proof}
    \begin{enumerate}
        \item The same arguments apply as in the proof of Property C2 in Proposition  \ref{prop:edelshort}.
        \item By definition, we have $\mathfrak{C}^{close}_{out}(v, G), \mathfrak{C}^{close}_{in}(v, G)\geq 0$, and for every node $v$ in an edgeless graph $\mathfrak{C}^{close}_{out}(v, G)= \mathfrak{C}^{close}_{in}(v, G)= 0$.
        \item  If $G$ of size $N\geq 2$ is weakly connected, we must have $Q(G)\geq 1/(N-1)$ for both the in- and out-closeness centrality since there is a node with an incoming, and a node with an outgoing edge. This proves that $l_N\geq 1/(N-1)$ for all $N\geq N_0$, where $N_0\in \mathbb{N}$, is necessary for P2.
        The super-spreader $v^\ast$ of a directed star graph $G^\ast$ of size $N\geq 2$ satisfies $\mathfrak{C}^{close}_{in}(v^\ast, G^\ast) = 0$ whereas $\mathfrak{C}^{close}_{in}(v, G^\ast) = 1/(N -1)$ for $v\in\mathcal{V}_{G^\ast}\setminus\{v^\ast\}$, which proves  $Q(G^\ast) = \max_{v\in\mathcal{V}_G}\mathfrak{C}^{close}_{in}(v, G^\ast)=1/(N-1)$. Analogously, considering a graph $\hat G$ obtained by inverting the direction of the edges of the directed star graph $G^\ast$, that is $\mathcal{V}_{\hat G}=\mathcal{V}_{G^\ast}$ and $\mathcal{E}_{\hat G}=\{(u,v)\mid (v,u)\in \mathcal{E}_{G^\ast}\}$,  we find $Q(\hat G) = \max_{v\in\mathcal{V}_G}\mathfrak{C}^{close}_{out}(v, \hat G)=1/(N-1)$. Since $G^\ast$ and $\hat G$ are weakly connected, sufficiency of there exists $N_0\in \mathbb{N}$ such that $l_N\geq 1/(N-1)$ for all $N\geq N_0$ for P2 follows. 
        \item         The super-spreader $v^\ast$ of a directed star graph $G^\ast$ of size $|\mathcal{V}_{G^\ast}|\geq 2$ satisfies  $\mathfrak{C}^{close}_{out}(v^\ast, G^\ast) =  1$. The same applies to a star node in a bidirectional star graph. 
        \item This follows from the proof of 3. 
    \end{enumerate}
\end{proof}

\subsubsection{DMFNR Based on Betweenness Centrality}

In this section we illustrate that centrality measures which define centrality of a node not in absolute terms but relative to the centrality of other nodes may not be suitable control for systemic (cyber) risk. To this end, recall the \text{betweenness centrality} given in \eqref{eq:betnode}.

\begin{proposition} Consider the acceptance set $\mathcal{A}$ in \eqref{eq:accset} with $Q(G) =\max_{v\in\mathcal{V}_G}\mathfrak{C}^{bet}(v,G)$. If $\mathcal{A}$ satisfies P1, then P3 is violated. Conversely, if $\mathcal{A}$ satisfies P3, then P1 is violated. 
\end{proposition}
\begin{proof}
If $G^c$ is a complete graph, then there is only one shortest path between two distinct nodes $u, w\in\mathcal{V}$, namely the edge between them. Therefore, in complete graphs we have $\sigma_{uw}(v) =0$ if $v\notin\{ u, w\}$, and we thus find $\mathfrak{C}^{bet}(v, G^c) = 0$ for all $v\in\mathcal{V}_{G^c}$. Hence, $Q(G^c)=0$. Clearly, any edge-less graph $G^\emptyset=(\mathcal{V}, \emptyset)$ also satisfies $Q(G^\emptyset)=0$. 
\end{proof}

\subsection{DMFNR Based on Stress Test Scenarios}\label{sec:DMFNRstress}

Suppose that the supervisor applies an (external) risk model operating on the supervised network. This model could, for instance, describe the evolution or spread of a risk throughout the network. Typical examples of such processes include bond or bootstrap percolation (\cite{Chiaradonna2023, Detering2019}), as well as continuous-time SIR or SIS epidemic models, as used in \cite{Awiszus2023b, Chernikova2022, Chernikova2023, xu2019cybersecurity}.
%
%
This external risk model is naturally employed to assess networks via stress-test simulations. Note that a macroprudential cyber stress testing framework for financial risk regulation has recently been proposed in \cite{Euro2020, Ros2020}, comprising four stages: scenario type (\textit{context}), initial cyber incident or attack (\textit{shock}), propagation through the network (\textit{amplification}), and the system-wide impact (\textit{systemic event}). The stages context, shock, and amplification define the dynamics of the cyber scenario. As these dynamics unfold, an initial shock results in specific outcomes, such as the total number of 
nodes affected by the cyber attack. The term \textit{systemic event} refers to outcomes that significantly disrupt the (financial) system's ability to perform critical functions. In order to define systemic events, the supervisor may establish an \textit{impact tolerance threshold}, as discussed in \cite{Ros2020}. For example, this threshold might be a critical fraction $\alpha \in (0,1)$, indicating the proportion of nodes affected during the incident.

Let $\mathcal{S}(G)$ denote the systemic event. Suppose that a network is deemed acceptable if the probability of such a systemic shock, $\mathbb{P}[\mathcal{S}(G)]$, does not exceed a specified threshold $\lambda \in (0,1)$. The corresponding acceptance set is 

\begin{equation}\label{eq:accset1}
    \mathcal{A} = \{ G \mid \mathbb{P}[\mathcal{S}(G)] \leq \lambda \}.
\end{equation}
Note that $\mathcal{A}$ is of type \eqref{eq:accset}, with $Q(G) = \mathbb{P}[\mathcal{S}(G)]$ and $l_{|\mathcal{V}_G|}\equiv \lambda$. 


\begin{remark}\label{rem:VaR}
Observe the similarity of \eqref{eq:accset1} to the acceptability criterion defining the Value-at-Risk (VaR) risk measure well-known in financial risk regulation, see, e.g., \cite[Section 4.4]{FS}. Notably, as in the case of the VaR, the acceptability criterion defining $\mathcal{A}$ might neglect tail risks.
\end{remark}

In the following we discuss a stress test based DMFNR where the external risk model is given by a SIR infection process.

\subsubsection{Stress Tests Based on an SIR Contagion Scenario} Consider Scenario 2 in our motivating example in the introduction, focusing on the spread of a cyber incident. The stress test begins by identifying \textit{initial shocks}, which can be modeled in various ways; for instance, infections might randomly occur at any network node.
In the \textit{amplification} phase, selecting an appropriate risk propagation model is essential. Cyber risks spread via interactions within the network. Epidemic models from mathematical biology provide a framework for such contagious risks. In the continuous-time SIR Markov model, which we will apply in this section, nodes are categorized into compartments: \emph{susceptible} $(S)$, \emph{infected} $(I)$, and \emph{recovered} $(R)$.
 For a network $G$ with $N$ enumerated nodes, the state of the process at time $t \geq 0$ is given by:
\begin{equation*}
X(t) = (X_1(t), \ldots, X_N(t)) \in E^N,
\end{equation*}
where $E = \{S, I, R\}$. State changes occur over time, starting at $t_0 = 0$ when the initial shock impacts the system.
In this model, a node may become infected by its infected in-neighbors, while recovery is independent of the states of other nodes in the network. Recovery corresponds to permanent immunity, meaning individuals do not face reinfection. The state transitions are described as follows:
\begin{align}
    \begin{split}\label{eq:SISSIRrates}
        X_i: S \rightarrow I &\quad \text{at rate} \quad \tau \sum_{j=1}^N a_{ji} \mathbbm{1}_{\{X_j(t) = I\}}, \\
        X_i: I \rightarrow R &\quad \text{at rate} \quad \gamma,
    \end{split}
\end{align}
where $\tau > 0$ is the infection rate and $\gamma > 0$ is the recovery rate.
\begin{figure}[h]
    \centering
    \begin{tikzpicture}[scale=1]
        \node[draw, shape=circle, fill=BrickRed!55!White] (1) at (0, 0) {I};
        \node[draw, shape=circle, fill=OliveGreen!55!White] (2) at (1, 0) {S};
        \node[draw, shape=circle, fill=BrickRed!55!White] (3) at (4, 0) {I};
        \node[draw, shape=circle, fill=BrickRed!55!White] (4) at (5, 0) {I};
        \draw[->] (1) -- (2);
        \draw[->] (3) -- (4);
        \draw[->, line width=0.5mm] (1.8, 0) -- node[anchor=south] {$\tau$} (3.2, 0);
        \node[draw, shape=circle, fill=BrickRed!55!White] (5) at (1, -1.5) {I};
        \node[draw, shape=circle, fill=White!85!gray] (6) at (4, -1.5) {R};
        \draw[->, decorate, decoration={snake, amplitude=.4mm, segment length=2mm, post length=1mm}, line width=0.5mm] (1.8, -1.5) -- node[anchor=south] {$\gamma$} (3.2, -1.5);
    \end{tikzpicture}
    \caption{Infection and recovery for the SIR network model.}
    \label{fig:SIR}
\end{figure}
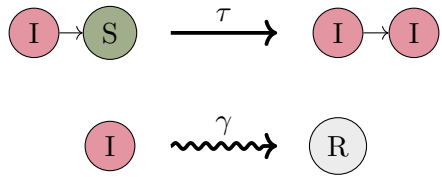

We define a systemic event $\mathcal{S}(G)$ in $G$ as the case where at least a fraction $\alpha\in (0, 1)$ of nodes becomes infected over the course of the epidemic. 

\begin{proposition}
Let $\mathcal{A}$ be as defined in \eqref{eq:accset1} with stress test described above based on the SIR contagion model. Then, for any non-empty intervention set $\mathcal{I}$ consisting of edge deletions $\mathcal{I} \subseteq \mathcal{I}_{\text{e\_del}}$,  the quantity $Q(G)=\mathbb{P}[\mathcal{S}(G)]$ is $\mathcal{I}$-monotone. Consequently, $(\mathcal{A}, \mathcal{I}, \mathcal{C})$ is a DMFNR whenever $\mathcal{C}$ is a cost function for $(\mathcal{A}, \mathcal{I})$.
\end{proposition}

\begin{proof}
The continuous-time SIR Markov model can be isomorphically represented as a random graph known as the \textit{epidemic percolation network} (EPN), as demonstrated in \cite{Kenah2007} and \cite[Section 6.3.3]{Kiss2017}. In particular, the distribution of final epidemic outcomes coincides  with the distribution of the combined out-components that originate from initially infected nodes in the EPN.
For a given network \(G \in \mathcal{G}\) the corresponding EPN \(\tilde{G}\) is constructed as follows:

\begin{enumerate}
    \item Initialize the vertex set \(\mathcal{V}_{\tilde{G}} = \mathcal{V}_G\) and set the edge set \(\mathcal{E}_{\tilde{G}} = \emptyset\).
    
    \item For each node \(v \in \mathcal{V}_G\), perform the following:
        \begin{enumerate}
            \item Sample the (potential) infection duration for node \(v\) from an exponential distribution with rate \(\gamma\). Given this duration, the transmissions from \(v\) to its neighbors occur independently.
            
            \item Choose an out-neighbor \(w \in \mathcal{N}_v^{G, out}\) and determine the transmission time from \(v\) to \(w\). This is done by sampling from an exponential distribution with rate \(\tau\).
            
            \item If the transmission time from \(v\) to \(w\) is less than the infection duration of \(v\), include the edge \((v, w)\) in the edge set \(\mathcal{E}_{\tilde{G}}\).
            
            \item Repeat steps (b) and (c) for all other out-neighbors of \(v\) in \(G\).
        \end{enumerate}
        
    \item The resulting EPN is then given by \(\tilde{G} = (\mathcal{V}_{\tilde{G}}, \mathcal{E}_{\tilde{G}})\).
\end{enumerate}
The \textit{out-component} $\mathcal{V}_{\tilde{G}^{out}_v}$ contains $v$ and all nodes $w\in\mathcal{V}_{\tilde G}\setminus\{ v\}$ for which there is a path from $v$ to $w$ in network $\tilde G$. These are those nodes that are eventually infected in case that $v$ is infected.
The final set of all infected nodes, \(\mathcal{V}_{G}^{fin}\), during the epidemic for an initially infected set \(\tilde{\mathcal{V}} \subseteq \mathcal{V}_G\), is represented by the union of their out-components in the EPN \(\tilde{G}\) :

\[
\mathcal{V}_{G}^{fin} = \bigcup_{v \in \tilde{\mathcal{V}}} \mathcal{V}_{\tilde{G}^{out}_v}.
\]

Clearly, the probability of a systemic outbreak event, defined as \(\vert\mathcal{V}_{G}^{fin}\vert\geq \alpha \vert\mathcal{V}_G\vert\), cannot be increased when an edge is removed from \(G\).
\end{proof}

\begin{remark} It can be shown that $Q(G)=\mathbb{P}[\mathcal{S}(G)]$ is not monotone for node splits. Indeed, this will happen under any contagion model, not only the SIR-model. Recall that node splits are related to diversification in the network. 
Thus, we observe another similarity to the VaR risk measure mentioned in Remark~\ref{rem:VaR} above, namely that diversification may be penalized in this risk model. 
\end{remark}

\begin{remark}
To analyze properties P1-P3, one must determine the probability distribution of outbreak sizes in the network under the given contagion model. This depends on the underlying model parameters, including the infection and recovery rates, as well as the distribution of the initial shock in the network. Moreover, closed forms of these distributions typically do not exist, necessitating approaches like Monte Carlo simulations. It is therefore a rather challenging exercise which is not in the scope of this paper.
\end{remark}

\begin{remark}(Alternative Interventions) 
In principle, in cases where an external risk model is applied, interventions other than those introduced in Section~\ref{sec:topinv} could be considered, namely interventions which target the parameters of the external model, such as the infection rate $\tau$ and the recovery rate $\gamma$ in the SIR model. Our framework does, in general, not assume such an external model. Therefore, we do not consider such interventions, even though they are conceivable in an extension of the decision-making framework. Such extensions are left to future research. Also note that, even if external risk models are applied, there are reasons to limit the interventions to the topological interventions we consider in this paper. 

Firstly, supervisors often have to consider various (external models for different) types of vulnerabilities, not just one. In that case it is natural to regulate the system's structure itself. 
The mentioned Zero Trust approach advocates for a fundamental robustification of the system architecture to become resistent to various threat scenarios.

Secondly, a standard approach to stress testing is to pre-calibrate the parameters of the external risk model to existing data, and then to leave  this model decribing a particular attack scenario unaltered while different network configurations (as a result of different protecting interventions) are stress-tested. In that case the interventions should not alter the parameters of the attack scenario.

Thirdly, in the case of the SIR model, for instance, the operational interpretation of altering parameters such as infection and recovery rate does not seem to be obvious. One would have to specify how real-world measures influence the parameters of the external model. 

Finally, inventions altering the parameters of the external risk model might not be enough. Consider the epidemic threshold of the continuous-time SIR model over an undirected network with countably infinite number of nodes, given by

\begin{equation}\label{eq:epi:thresh}
    \frac{\tau}{\tau + \gamma} \frac{\mathbb{E}[K_G^2 - K_G]}{\mathbb{E}[K_G]} > 1,
\end{equation}
 see \cite[Chapter 10]{Barabasi2016}, \cite[Chapter 6.2]{Kiss2017}, and \cite[Section V.B.4]{PastorSatorras2015}. Here $K_G$ is (a random variable with) the degree distribution of $G$.
If \eqref{eq:epi:thresh} is satisfied, then epidemic outbreaks are possible; otherwise, they are not. We can consider the infinite network a proxy for very large real-world networks, such as the internet, which often exhibit significant heterogeneity. These networks typically feature a few hubs and many sparsely connected nodes. These so-called \textit{scale-free networks} often follow a power-law degree distribution, given by $\mathbb{P}(K_G = k) \sim k^{-\alpha}$. As the network size increases, the second moment of this degree distribution diverges, ensuring that \eqref{eq:epi:thresh} is satisfied regardless of the values of the infection and recovery parameters $\tau$ and $\gamma$. Consequently, with increasing network size epidemic outbreaks remain possible despite changes to these parameters. However, topological interventions aim to modify the topology and, therefore, the degree distribution of the network, potentially avoiding the conditions set by \eqref{eq:epi:thresh}.
\end{remark}

\section{Outlook}\label{sec:outlook}

In the following we collect a few of open questions and challenges for future research:
\paragraph{Calibration to Real World Networks} While this study provides the theoretical foundation and also presents first stylized examples of decision-making frameworks for networks resilience (DMFNR), the next step is to apply such DMFNRs to managing real world networks, such as the cyber mapping mentioned in the motivating example in the introduction. In this context particularly suitable DMFNRs will be identified together with important parameters such as a desired minimal levels of network functionality, the border between acceptability and non-acceptability in practice, etc.  

\paragraph{Computational Challenges} { Computational problems stem primarily from the high dimensionality of the networks. Determining suitable or even optimal network manipulations may therefore be challenging.  
    A possible solution may be to develop new machine learning-based algorithms to identify good, not necessarily strictly optimal, ways to secure some given network.

\paragraph{Model Uncertainty} A major problem in practice is that the supervisor might only possess incomplete information about the network she has to secure. 

  \paragraph{Weighted Networks} As mentioned in Remark~\ref{rem:weighted}, network models may benefit from adding edge weights. For example, in models of financial systems edge weights represent the liabilities between financial institutions. Default contagion occurs when a financial institution's failure to meet its liabilities—--represented by its outgoing edges to other institutions—--triggers a cascade of defaults among connected institutions, see  \cite{Detering2019b}. 
  The extension of the presented framework to weighted graphs would, however, require a number of modifications, for instance as regards the interventions which now also target the edge weights.

}


    \paragraph{Area of Application} Even though, on the most basic level, this paper introduces decision-making frameworks for network resilience for any type of network, the later discussion is tailored to systemic cyber risk. It would be interesting to explore applications to other examples of critical infrastructure networks such as power grids, transportation systems, and production networks.

\clearpage

\renewcommand{\thesection}{A}

\section{Basic Notions for Graphs}\label{sec:graph:basics}

\subsection{Adjacency Matrix}\label{sec:adjacency}
For a network of size $N$, the nodes can be enumerated, say as $v_1,\cdots , v_N$. Given an enumeration of the vertex set $\mathcal{V}_G$, the \textit{network topology}, i.e., the spatial pattern of interconnections, of a network $G$ is described by its \textit{adjacency matrix} 
$A_G = (A_G (i,j))_{i,j=1,\cdots , N}\in\{ 0,1\}^{N\times N}$, where
\begin{equation*}
    A_G (i,j) = \begin{cases}
           1,& \text{if } (v_i,v_j)\in\mathcal{E}_G\\
           0,& \text{else}.
    \end{cases}
\end{equation*}
Note that $G$ is undirected if and only if $A_G$ is a symmetric matrix.

\subsection{Neighborhoods of Nodes and Node Degrees}\label{sec:neighbor}

For a graph $G=(\mathcal{V}_G, \mathcal{E}_G)$ and node $v\in\mathcal{V}_G$, we define $\mathcal{N}^{G, in}_v :=\{w\in\mathcal{V}_G\vert (w,v)\in\mathcal{E}_G\}$ and $\mathcal{N}^{G, out}_v :=\{w\in\mathcal{V}_G\vert (v,w)\in\mathcal{E}_G\}$ as the \textit{in-} and \textit{out-neighborhoods} of  node $v$. 
The {\em neighborhood} is $\mathcal{N}^{G}_v:=\mathcal{N}^{G, in}_v\cup \mathcal{N}^{G, out}_v$. A network $G$  is undirected if and only if $\mathcal{N}^{G, in}_v = \mathcal{N}^{G, out}_v = \mathcal{N}^G_v$ for all $v\in\mathcal{V}_G$. 

The \textit{incoming} or \textit{in-degree}, i.e., the number of edges arriving at node $v$, is defined as $k^{G, in}_v = \vert\mathcal{N}^{G, in}_v\vert$, and the \textit{outgoing} or \textit{out-degree} as $k^{G, out}_v = \vert\mathcal{N}^{G, out}_v\vert$. 
The \textit{(total) degree} $k_v^G$ of node $v$ in network $G$ is defined as 
\begin{equation*}
    k_v^G = \frac12(k^{G, in}_v+k^{G, out}_v ).
\end{equation*}
Typically, total degrees are studied in the context of undirected networks $G$ where we find that $k_v^G = k_v^{G, in} =k_v^{G, out}$ due to the fact that $\mathcal{N}^{G, out}_v=\mathcal{N}^{G, in}_v$ for all $v\in \mathcal{V}_G$.

\subsection{Walks, Paths, and Connectivity of Graphs}\label{sec:connectivity}
Given a graph $G$, a \textit{walk} of length $n$ from node $v$ to node $w$ in $G$ is an $(n+1)$-tuple of nodes $(v_1,v_2,\ldots,v_{n+1})$ 
such that $(v_i, v_{i+1})\in\mathcal{E}_G$ for all $1\leq i\leq n$ and $v_1=v$ and $v_{n+1}=w$. A \textit{path} from $v$ to $w$ is a walk where all nodes are distinct.

A graph $G$ is called \textit{strongly connected} if for each node pair $v, w\in\mathcal{V}_G$ there exist paths from $v$ to $w$ and from $w$ to $v$.  
A directed graph \(G\) is called \textit{weakly connected} if its undirected version, that is the graph $\tilde{G} = (\mathcal{V}_G, \mathcal{E}_{\tilde{G}})$ with $\mathcal{E}_{\tilde{G}} = \{ (v,w)\vert (v,w)\in\mathcal{E}_G \text{ or } (w,v)\in\mathcal{E}_G\}$, is strongly connected. A \textit{(weakly connected) component} $\tilde{G}$ of network $G$ is a (weakly) connected subgraph, and two components $G_1 = (\mathcal{V}_1, \mathcal{E}_1)$, $G_2 = (\mathcal{V}_1, \mathcal{E}_2)$ of $G$ are called \textit{disconnected} if there are no $v\in\mathcal{V}_1$ and $w\in\mathcal{V}_2$ such that there is a path from $v$ to $w$ or from $w$ to $v$ in $G$. 


\renewcommand{\thesection}{B}
\section{Undirected Networks}\label{app:undir}

We denote the set of undirected networks by \[\mathcal{G}^{ud}:=\{ G = (\mathcal{V}_G, \mathcal{E}_G)\in\mathcal{G}\vert (v,w)\in\mathcal{E}_G \Leftrightarrow (w,v)\in\mathcal{E}_G\}\subset \mathcal{G}.\] The pair $\{(v,w),(w,v)\}\in \mathcal{E}_G$ will be referred to as a \textit{full edge} of the undirected graph $G = (\mathcal{V}_G, \mathcal{E}_G)$.
In the following, we discuss how the presented theory in the main part of this paper would change if we restrict the domain to $\mathcal{G}^{ud}$ instead of $\mathcal{G}$. Indeed, most of the  definitions and results can easily be adapted to this class of networks and we leave this to the reader apart from a few comments and additional examples which are collected in this section. Regarding the Properties P2 and P3 presented in Section~\ref{sec:acceptance}, note that an undirected graph is weakly connected if and only if it is strongly connected, which in this case is only referred to as being connected. Moreover, for an arbitrary graph $G\in\mathcal{G}^{ud}$, any super-spreader $v^\ast\in\mathcal{V}_G$ is a star node. 
We thus obtain the following versions of P2 and P3 for undirected graphs: 
\begin{itemize}
    \item P2:  There exists $N_0\in \mathbb{N}$ such that there is a connected graph in $\mathcal{A}$ of size $N$ for all $N\geq N_0$.
    
    \item P3: Any network with a star node is not acceptable.
\end{itemize}

\subsection{Interventions for Undirected Networks}\label{sec:undir}
If we restrict the discussion to the class $\mathcal{G}^{ud}$, then any undirected graph should always be transformed into a new undirected graph. For convenience, as we view the undirected networks $\mathcal{G}^{ud}$ as a subset of the potentially directed networks $\mathcal{G}$,  the following elementary interventions on undirected graphs are defined on $\mathcal{G}$, that is $\kappa:\mathcal{G}\to\mathcal{G}$, but satisfy $\kappa(\mathcal{G}^{ud})\subseteq \mathcal{G}^{ud}$. Note that only those interventions from Section~\ref{sec:topinv} that come with a manipulation of network edges need to be adjusted. 

\begin{itemize}
    \item[$\mathcal{I}^{ud}_{e\_ del}$] \textit{Edge Deletion}: Consider a node tuple  $v,w\in\mathbb{V}$. We let
\begin{equation*}
    \kappa^{v,w}_{e\_ del} := \kappa^{(w,v)}_{e\_ del}\circ\kappa^{(v,w)}_{e\_ del}:
    G\mapsto 
    (\mathcal{V}_G, \mathcal{E}_G\setminus\{ (v,w), (w,v)\})
\end{equation*}
denote the deletion of the full edge between nodes $v$ and $w$. Clearly, $\kappa^{v,w}_{e\_ del}=\kappa^{w,v}_{e\_ del}$. We set $$\mathcal{I}^{ud}_{e\_ del}:=\{\kappa^{v,w}_{e\_ del}\mid v,w\in \mathbb{V}\}.$$
\item[$\mathcal{I}^{ud}_{e\_ add}$] \textit{Edge Addition}: The addition of an edge between nodes $v,w\in\mathbb{V}$ with $v\neq w$ is given by
 \[\kappa^{v,w}_{e\_ add} := \kappa^{(w,v)}_{e\_ add}\circ\kappa^{(v,w)}_{e\_ add}:  G\mapsto
    (\mathcal{V}_G, \mathcal{E}_G\cup (\{ (v,w), (w,v)\}\cap \mathcal{V}_G\times\mathcal{V}_G)),\] and we let $$\mathcal{I}^{ud}_{e\_ add}:=\{\kappa^{v,w}_{e\_ add} \mid v,w\in \mathbb{V}\}.$$
\item[$\mathcal{I}^{ud}_{shift}$] \textit{Edge Shift}: An existing full edge between $v$ and $w$ can be shifted to a full edge between $q$ and $r$ by 
$$ \kappa_{shift}^{\{v,w\}, \{q,r\}}:= \kappa_{shift}^{(w.v), (r,q)}\circ \kappa_{shift}^{(v,w), (q,r)}.$$
  We set $$\mathcal{I}^{ud}_{shift}:=\{\kappa_{shift}^{\{v,w\}, \{q,r\}}\mid v,w,q,r\in \mathbb{V}\}.$$
  
\item[$\mathcal{I}^{ud}_{split}$] \textit{Node Splitting}: For node splitting in undirected networks, we need to restrict to interventions from the set
$$ \mathcal{I}^{ud}_{split} := \{  \kappa_{split}^{\mathcal{L}, v, \tilde{v}}\vert v, \tilde{v}\in\mathcal{V}, \mathcal{L}\subset\mathbb{E}, (q,r)\in\mathcal{L}\Leftrightarrow (r,q)\in\mathcal{L}\}\subset \mathcal{I}_{split}.$$

\end{itemize}

\subsection{Examples}
As in the main part of this paper, we consider examples based on acceptance sets $\mathcal{A}\subset\mathcal{G}^{ud}$ of the form
\begin{equation}\label{eq:accsetundir}
\mathcal{A} = \{ G \in \mathcal{G}^{ud}\vert Q(G)\leq l_{|\mathcal{V}_G|}\}
\end{equation}
where $Q:\mathcal{G}^{ud}\to\mathbb{R}\cup \{-\infty, \infty\}$ and $l_N\in \mathbb{R}$ for all $N\in \mathbb{N}$. In principle, any $Q$ from the previous discussion on directed networks can also be applied here and the results from the directed case essentially also apply to the undirected case. Only some proofs or bounds $l_N$ need to be adjusted as we will see in the following. 

\subsubsection{DMFNR Based on the Average Total Degree} For undirected networks we consider the distribution of the total degree $k_v^G$, see Section~\ref{sec:neighbor}, which, of course, equals the in- and out-degree distributions. Let $K_G$ be a random variable, the law of which coincides with the total degree distribution. 
Similar to the directed case, choosing $Q(G) =\mathbb{E}\big[K_G\big]$ in \eqref{eq:accsetundir} yields a DMFNR but $\mathcal{A}$ does not satisfy P2 and P3 simultaneously. Indeed recall Proposition~\ref{prop:firstmom} and compare to Proposition~\ref{prop:firstmomundir} below. In this case the proof is based on \textit{tree graphs}:  a connected undirected acyclic graph $G^T\in\mathcal{G}^{ud}$ is called an \textit{undirected tree}.
\begin{figure}[h]
\begin{center}
		\begin{minipage}[t]{0.49\linewidth}
			\centering
	{\includegraphics[trim={2cm 2.5cm 2 2cm},clip,width=1.2\textwidth]{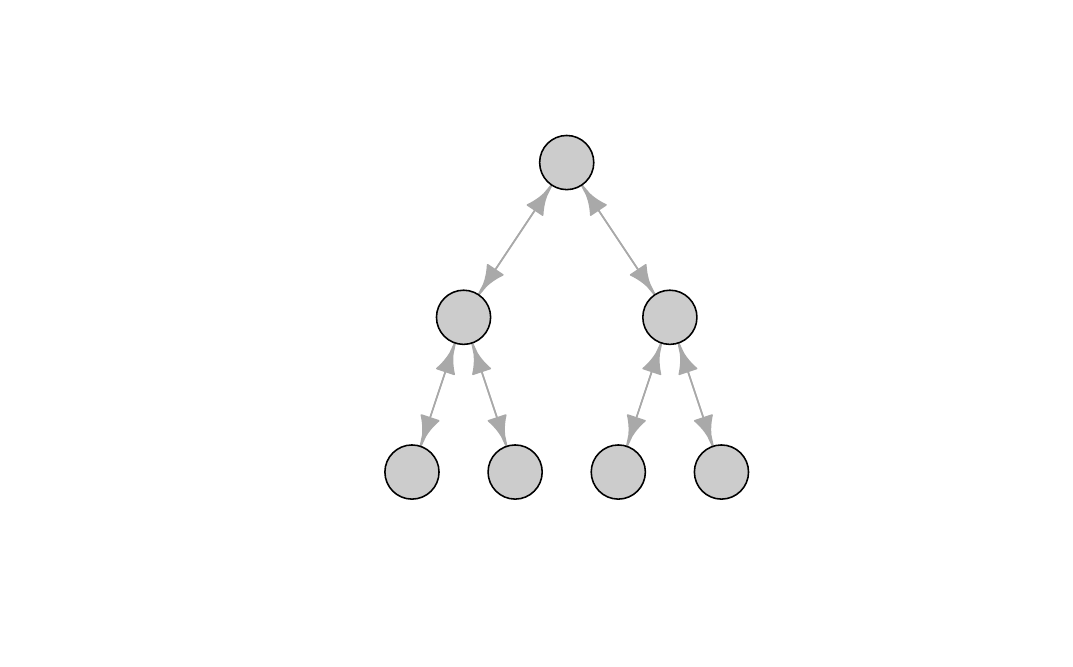}}
			
		\end{minipage}
    		\begin{minipage}[t]{0.49\linewidth}
			\centering
	{\includegraphics[trim={1cm 5cm 1cm 5cm},clip,width=0.55\textwidth]{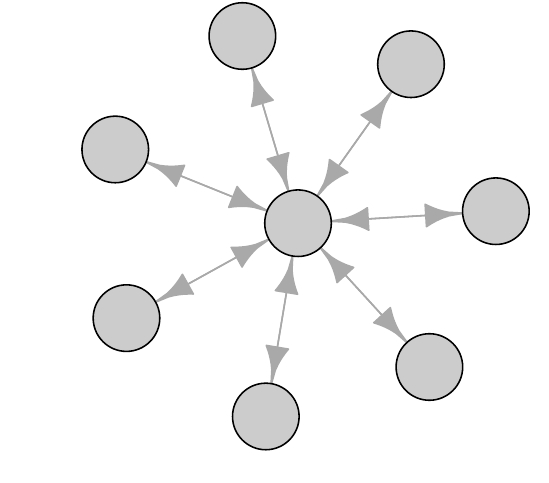}}
		\end{minipage}
		\end{center}
		\caption{An undirected tree graph (left) and undirected star graph (right) consisting of $N=7$  nodes. Note that both the undirected line graph from Figure \ref{fig:ring} and the undirected star graph are special cases of a tree since they do not contain any cycles.}
		\label{fig:treeline}
\end{figure}

\begin{proposition}\label{prop:firstmomundir}
Suppose that $\mathcal{A}$ is given by  \eqref{eq:accsetundir} with $Q(G) = \mathbb{E}\big[K_G\big]$. If P2 is satisfied, then P3 is violated, and vice versa.
\end{proposition}
\begin{proof}
An undirected connected graph $G$ of size $N$ must contain a spanning tree $\tilde{G}$, see Corollary 7 in \cite{Bollobas1998}. Moreover, since $Q$ is $\mathcal{I}_{e\_ del}$-monotone, we have $Q(\tilde{G})\leq Q(G)$. However, $\tilde{G}$ has the same number of edges as any undirected star graph $G^\ast$ of the same size since both are undirected tree graphs, see \cite[Corollary 8]{Bollobas1998}, and therefore $Q(G)\geq Q(\tilde{G}) = Q(G^\ast) = 2(N-1)/N$. This proves the assertion.
\end{proof}

\subsubsection{DMFNR Based on the Second Moment of the Degree Distribution}
As in the directed case, we find that choosing $Q(G) =\mathbb{E}\big[ K_G^2\big]$ in \eqref{eq:accsetundir} is a suitable way to define acceptance sets $\mathcal{A}$ which satisfy properties P1-P4. For the proof, we need the following lemma: 

\begin{lemma}\label{lem:treeline}
Fix $N\in \mathbb{N}$. The minimal second moment $Q(G^T) =\mathbb{E}\big[K_{G^T}^2\big]$ among all undirected tree graphs $G^T\in\mathcal{G}^{ud}$ of size $N$ is attained by the undirected line graphs $G^L =(\mathcal{V}_{G^L}, \mathcal{E}_{G^L})$, where for $\mathcal{V}_{G^L} = \{ v_1,\cdots, v_N\}\subset\mathbb{V}$ we have $\mathcal{E}_{G^L} = \{ (v_1, v_2), (v_2, v_1),$\\ $\cdots ,(v_{N-1}, v_N), (v_N, v_{N-1})\}$.
\end{lemma}
\begin{proof}
We prove the result by induction. First, note that for $N=2$ and some arbitrary $v_1, v_2\in\mathbb{V}$, the only undirected tree $G^T$ is given by $G^T =(\{v_1, v_2\}, \{(v_1, v_2), (v_2, v_1)\})$, which is a line graph. Now, suppose that the statement holds for some $N\geq 2$, and consider a tree $G^T$ of size $N+1$. Then $G^T$ contains at least two nodes with a total degree of 1, called \textit{leaves}, see Corollary 9 in \cite{Bollobas1998}. Consider an enumeration of the nodes $\mathcal{V}_{G^T}$ such that $v_{N+1}$ is a leaf of $G^T$ with (only) neighbor $v_N$. Delete node $v_{N+1}$ such that we obtain the network $H =\kappa_{n\_ del}^{v_{N+1}}(G^T)$. Note that $H$ is a tree graph of size $N$. Now consider the line graph $G^L$ of the same size defined on the vertex set $\mathcal{V}_{G^T}\setminus\{v_{N+1}\}$ with leaf $v_N$. By induction hypothesis, we have
 \begin{equation*}
     \mathbb{E}\big[ K_{G^L}^2\big]\leq \mathbb{E}\big[ K_H^2\big],\quad \text{ and thus }\quad  \sum_{i=1}^{N}\big( k_{v_i}^{G^L}\big)^2\leq  \sum_{i=1}^{N}\big( k_{v_i}^H\big)^2.
 \end{equation*}
 Now, let $\tilde{G^L} = (\mathcal{V}_{G^T}, \mathcal{E}_{G^L}\cup \{ (v_N, v_{N+1}), (v_{N+1}, v_N)\})$, which is a line graph of size $N+1$. Since node $v_N$ is a leave in $G^L$, we find $1=k_{v_N}^{G^L}\leq k_{v_N}^H$. Moreover,  $k_{v_i}^{\tilde{G^L}} = k_{v_i}^{G^L}$, $k_{v_i}^{G^T} = k_{v_i}^H$ for $i=1,\cdots, N-1$, and we have $k_{v_N}^{\tilde{G^L}} = k_{v_N}^{G^L}+1=2$, $k_{v_N}^{G^T} = k_{v_N}^{H}+1\geq 2$, and $k_{v_{N+1}}^{G^T} = k_{v_{N+1}}^{\tilde{G^L}} =1$. In total, this yields
\begin{align*}
    \mathbb{E}\big[ K_{\tilde{G^L}}^2\big] &= \frac{1}{N+1} \sum_{i=1}^{N+1}\big( k_{v_i}^{\tilde{G^L}}\big)^2 = \frac{1}{N+1} \Big(  \sum_{i=1}^{N-1}\big( k_{v_i}^{G^L}\big)^2 + \big(k_{v_N}^{G^L} +1\big)^2 + 1^2\Big) \\
    &= \frac{1}{N+1} \Big(  \sum_{i=1}^{N}\big( k_{v_i}^{G^L}\big)^2 + 2k_{v_N}^{G^L} + 1 + 1\Big) = \frac{1}{N+1} \Big(  \sum_{i=1}^{N}\big( k_{v_i}^{G^L}\big)^2 + 3 + 1\Big)\\
    &\leq \frac{1}{N+1} \Big(  \sum_{i=1}^{N}\big( k_{v_i}^{H}\big)^2 + 3 + 1\Big) =   \frac{1}{N+1} \Big(  \sum_{i=1}^{N-1}\big( k_{v_i}^{H}\big)^2 + \big(k_{v_N}^{G^T}-1\big)^2 + 3 + 1\Big) \\
    &= \frac{1}{N+1} \Big(  \sum_{i=1}^{N}\big( k_{v_i}^{G^T}\big)^2 -2 k_{v_N}^{G^T} + 1 + 3 + 1\Big) \leq \frac{1}{N+1} \Big(  \sum_{i=1}^{N}\big( k_{v_i}^{G^T}\big)^2 + 1\Big) = \mathbb{E}\big[ K_{G^T}^2\big].
 \end{align*}
\end{proof}

\begin{proposition}\label{prop:secundir}
Suppose that $\mathcal{A}$ is given by \eqref{eq:accsetundir} with $Q(G) = \mathbb{E}\big[ (K_G^2\big]$.
\begin{enumerate}
 \item $Q$ is $\mathcal{I}$-monotone for any $\mathcal{I}\subset \mathcal{I}^{ud}_{e\_del}\cup \mathcal{I}^{ud}_{split}$. Hence $(\mathcal{A}, \mathcal{I}, \mathcal{C})$ is a DMFNR for any cost function for $(\mathcal{A}, \mathcal{I})$. 
\item $\mathcal{A}$ satisfies P1 if and only if $l_N\geq 0$ for all $N\geq 2$,
 \item $\mathcal{A}$ satisfies  P2 if and only if there is a $N_0$ such that $l_N\geq 4- 6/N$ for all $N\geq N_0$.
 \item $\mathcal{A}$ satisfies  P3 if and only if $l_N < N-1$ for all $N\geq 1$.
\end{enumerate}
\end{proposition}
\begin{proof}
1. and 2. are straightforward modifications of the proof of the corresponding results in Propositions~\ref{prop:secmomint} and \ref{prop:secout}. \\
3.: Note that  since $Q(G)=\mathbb{E}\big[K_G^2\big]$ is $\mathcal{I}_{e\_ del}$-monotone according to 1., and every connected network contains a spanning tree, it suffices to determine the minimizer of $Q$ among all undirected tree graphs. The second moment of total degrees is minimized among all undirected tree graphs of a fixed size $N$ by the undirected line graphs, see Lemma \ref{lem:treeline}, and can be calculated for $N\geq 2$ as
 \begin{equation*}
     \mathbb{E}\big[ K_{L}^2\big] = \frac{1}{N} \big( 2\cdot 1 + (N-2)\cdot 4\big) = \frac{4N-6}{N} = 4-\frac{6}{N}.
 \end{equation*}
4.: For the undirected star graph $G^\ast$ of size $N$ we obtain
    \begin{equation*}
        \mathbb{E}\big[ (K_{G^\ast})^2\big] = \frac{1}{N} \Big( 1\cdot (N-1)^2 + (N-1)\cdot 1^2\Big) =\frac{1}{N} \big((N-1)^2 + (N-1)\big) = N-1.  
   \end{equation*}
\end{proof}

\subsubsection{DMFNR Based on Epidemic Threshold}\label{sec:epi:thresh}
Recall the epidemic threshold of the SIR model for undirected networks given in \eqref{eq:epi:thresh}. 
Rearranging the inequality shows that the ratio of second and first moment may be a promising candidate for the control of network contagion 
\begin{equation}\label{eq:epi:Q}Q(G)= \begin{cases} \frac{\mathbb{E}\big[ K_G^2\big]}{\mathbb{E}\big[ K_G\big]} & \mbox{if}\;  \mathbb{E}\big[ K_G\big]>0 \\ 0 & \mbox{if}\; \mathbb{E}\big[ K_G\big]=0. \end{cases} \end{equation}  
Note, however, that risk management with respect to this quantity is not completely compatible with edge deletions:

\begin{proposition}\label{prop:epi:nm}
    $Q$ as given in \eqref{eq:epi:Q} is not $\mathcal{I}^{ud}_{e\_ del}$-monotone.
\end{proposition}
\begin{proof}
     Suppose we delete the edges $\{(v, w), (w, v)\}$ between two adjacent nodes $v, w\in\mathcal{V}_G$ in a network $G$ of size $N\geq 3$, and let $\tilde{G} = \kappa^{v, w}_{e\_ del}(G)$. Further suppose that $\tilde{G}$ does contain edges, so that $\mathbb{E}\big[K_{\tilde{G}}\big] , \mathbb{E}\big[K_{\tilde{G}}^2\big]>0 $ and thus $Q(\tilde{G})=\mathbb{E}\big[K_{\tilde{G}}^2\big]/\mathbb{E}\big[K_{\tilde{G}}\big]$. If $Q$ were $\mathcal{I}^{ud}_{e\_ del}$-monotone, then 
\begin{equation}\label{eq:eddelrat}
    \frac{\mathbb{E}\big[K_G^2\big]}{\mathbb{E}\big[ K_G\big]} -  \frac{\mathbb{E}\big[K_{\tilde{G}}^2\big]}{\mathbb{E}\big[K_{\tilde{G}}\big]} \geq 0, \quad \mbox{i.e.}\; \mathbb{E}\big[K_{\tilde{G}}\big] \mathbb{E}\big[K_G^2\big] - \mathbb{E}\big[K_G\big] \mathbb{E}\big[ K_{\tilde{G}}^2\big] \geq 0.
\end{equation}
We can express the moments of $\tilde{G}$ in terms of the moments of $G$ by
\begin{equation*}
     \mathbb{E}\big[K_{\tilde{G}}\big] = \mathbb{E}\big[ K_G\big] -\frac{2}{N},\qquad  \mathbb{E}\big[ K_{\tilde{G}}^2\big] = \mathbb{E}\big[ K_G^2\big] + \frac{2}{N} (1- (k_v^G +k_w^G )).
\end{equation*}
Therefore, 
\begin{align*}
    \mathbb{E}\big[K_{\tilde{G}}\big] \mathbb{E}\big[ K_G^2\big]- \mathbb{E}\big[ K_G\big] \mathbb{E}\big[ K_{\tilde{G}}^2\big] &= \Big(\mathbb{E}\big[ K_G\big] - \frac{2}{N}\Big) \mathbb{E}\big[ K_G^2\big] - \mathbb{E}\big[ K_G\big] \Big( \mathbb{E}\big[ K_G^2\big] + \frac{2}{N} \big( 1 - (k_v^G + k_w^G)\big)\Big) \\
    &= \frac{2}{N} \Big( -\mathbb{E}\big[ K_G^2\big] -\mathbb{E}\big[ K_G\big] + \mathbb{E}\big[ K_G\big] \big( k_v^G + k_w^G\big)\Big).
\end{align*}
Hence, \eqref{eq:eddelrat} is satisfied if and only if $\mathbb{E}\big[ K_G\big] \big(k_v^G + k_w^G\big) - \big(\mathbb{E}\big[ K_G^2\big] + \mathbb{E}\big[ K_G\big] \big)\geq 0$, i.e., when 
\begin{equation}\label{eq:ratiocond}
    \big (k_v^G + k_w^G\big) \geq \frac{\mathbb{E}\big[ K_G^2\big]}{\mathbb{E}\big[ K_G\big]} +1.
\end{equation} One easily constructs examples where \eqref{eq:ratiocond} is not satisfied.  
\end{proof}
For the management of risk under edge deletions in a given network $G$, we thus need to restrict ourselves to those edge deletions that satisfy \eqref{eq:ratiocond}, i.e., that target edges between nodes with sufficiently large degrees.  


\begin{proposition}\label{prop:epithreshundir}
Let $\mathcal{A}$ be given by \eqref{eq:accsetundir} and $Q$ in \eqref{eq:epi:Q}. Then 
\begin{enumerate}
\item $Q$ is $\mathcal{I}$-monotone for any $\mathcal{I}\subset\mathcal{I}^{ud}_{split}$, and $(\mathcal{A}, \mathcal{I}, \mathcal{C})$ is a DMFNR for any cost function for $(\mathcal{A}, \mathcal{I})$. 
\item $\mathcal{A}$ satisfies P1 if and only if $l_N\geq 0$ for all $N\geq 2$.
 \item $\mathcal{A}$ satisfies  P2 if and only if there is a $N_0$ such that $l_N\geq 2-\frac{1}{N-1}$ for all $N\geq N_0$.
 \item Suppose that $l_1<0$, $l_N < N/2$ for all $2\leq N\leq 6$ and that $l_N< \frac{4N-1}{N+1}$ for all $N\geq 7$. Then $\mathcal{A}$ satisfies  P3. 
\end{enumerate}
\end{proposition}
\begin{proof}
    \begin{enumerate}
        \item Let $v\in\mathcal{V}_G$ be a node of the network $G$ which is split into $v$ and $\tilde{v}$ with a resulting network $\tilde{G}$. Again, utilizing \eqref{eq:ndeg} and $(N+1) \mathbb{E}\big[ K_{\tilde{G}}\big] = N \mathbb{E}\big[ K_G\big]$, we see that
    \begin{equation*}
     \frac{\mathbb{E}\big[ K_{\tilde{G}}^2\big]}{\mathbb{E}\big[ K_{\tilde{G}}\big]} =  \frac{\frac{1}{N+1} \Big( \sum_{w\neq v, \tilde{v}} (k_{w}^{\tilde{G}})^2 + (k_{v}^{\tilde{G}})^2 + (k_{\tilde{v}}^{\tilde{G}})^2\Big)}{ \frac{1}{N+1} \Big( \sum_{w\neq v, \tilde{v}} k_{w}^{\tilde{G}} + k_{v}^{\tilde{G}} + k_{\tilde{v}}^{\tilde{G}}\Big)}   \leq \frac{\sum_{w\in\mathcal{V}_G} \big(k_w^G\big)^2}{\sum_{w\in\mathcal{V}_G} k_w^G} = \frac{\mathbb{E}\big[ K_G^2\big]}{\mathbb{E}\big[ K_G\big]}.
    \end{equation*}
    \item is obvious.
    \item For a  undirected line graph $G^L$ of size $N\geq 2$, we obtain $$Q(G^L)=\frac{\mathbb{E}\big[K_{G^L}^2\big]}{\mathbb{E}\big[ K_{G^L}\big]} = \frac{4-\frac6N}{2-\frac{2}{N}} = \frac{2N-3}{N-1}.$$ We show that for any connected graph $H\in \mathcal{G}^{ud}$ of size $N\geq 2$ we have \begin{equation}\label{eq:help:epi}Q(H)\geq \frac{2N-3}{N-1}= 2-\frac{1}{N-1}.\end{equation} To this end, consider the following three cases: If $k_v^H\geq 2$ for all $v\in \mathcal{V}_H$, it follows that $(k_v^H)^2\geq 2k_v^H$ for all $v\in \mathcal{V}_H$ and thus $$\mathbb{E}\big[K_{H}^2\big]\geq 2 \mathbb{E}\big[K_{H}\big]\geq \left(2-\frac{1}{N-1}\right) \mathbb{E}\big[K_{H}\big],$$ so \eqref{eq:help:epi} holds. Suppose that $k_v^H\leq 2$ for all $v\in \mathcal{V}_H$. Then $H$ is either a undirected ring graph (i.e.\ $k_v^H= 2$ for all $v\in \mathcal{V}_H$) or an undirected line graph. In case of the undirected ring graph choose arbitrary $v,w\in \mathcal{V}_H$ such that $(v,w)\in \mathcal{E}_H$ (and thus also $(w,v)\in \mathcal{E}_H$). Note that $v,w$ satisfy condition \eqref{eq:ratiocond} in the proof of Proposition~\ref{prop:epi:nm}, so that deleting the edges $(v,w)$ and $(w,v)$ decreases $Q$. Notice that after the deletion of those edges we are left with an undirected line graph. As a last case, suppose that there exists $v,w\in \mathcal{V}_H$ such that $k_v^H\geq 3$ and $k_w^H=1$. Choose a neighbor $s\in \mathcal{V}_H$ of $v$ such that 
    \begin{itemize}
      \item $s\neq w$,
      \item $s$ is not a neighbor of $w$,
      \item there is a path from $v$ to $w$ not passing through the edge $(v,s)$.
    \end{itemize}
    This is possible because $k_v^H\geq 3$. Indeed, if $w$ happens to be a neighbor of $v$, then choose as $s$ any of the other neighbors of $v$. If $w$ is not a neighbor of $v$ and there is a neighbor $u$ of $v$ such that any path from $u$ to $w$ passes through the node $v$, then let $s=u$. Finally, if $w$ is not a neighbor of $v$ and all neighbors $u$ of $v$ allow for a path from $u$ to $w$ which does not pass through $v$, then at most one of those neighbors can be a neighbor of $w$ ($k_w^H=1$), so let $s$ be one of the other neighbors. Now let $\tilde H$ denote the graph obtained from $H$ by removing the edges $(s,v)$ and $(v,s)$ and adding the edges $(s,w)$ and $(w,s)$. Note that $\tilde H$ is connected, because any path through $(v,s)$ or $(s,v)$ can be redirected to a path passing through $w$. Then $\mathbb{E}\big[K_{H}\big]=\mathbb{E}\big[K_{\tilde H}\big]$ since we did not alter the total number of edges. However, \begin{eqnarray*}\mathbb{E}\big[K_{\tilde H}^2\big] &= &\frac1N \left(\sum_{\substack{u\in \mathcal{V}_H\\ u\neq v,w}}(k_u^H)^2 + (k_v^H-1)^2 + (k_w^H+1)^2\right) \\ &\leq &\frac1N \left(\sum_{\substack{u\in \mathcal{V}_H\\ u\neq v,w}}(k_u^H)^2 + (k_v^H)^2 + (k_w^H)^2\right) \;= \; \mathbb{E}\big[K_{H}^2\big]\end{eqnarray*}
    since $x^2+y^2\geq (x-1)^2 + (y+1)^2$ whenever $x-y\geq 1$. 
    Consequently, $Q(\tilde H)\leq Q(H)$. If $\tilde H$ falls under one of the first two cases, the assertion is proved. Otherwise, $\tilde H$ itself falls under the third case and can again be altered accordingly, with decreasing $Q$, until we finally satisfy the conditions of one of the first two cases.  
\item Suppose that $G\in \mathcal{G}^{ud}$ has $N\geq 2$ nodes and contains a star node $v^\ast\in \mathcal{V}_G$. Let $\tilde G=\kappa_{n\_ del}^{v^\ast}(G)$. We estimate 
\begin{equation*}
Q(G) = \frac{\mathbb{E}\left[ K_{G}^2\right]}{\mathbb{E}\left[ K_{G}\right]}=  \frac{N+ 2\mathbb{E}\left[ K_{\tilde{G}}\right]+ \mathbb{E}\left[ K_{\tilde{G}}^2\right]}{2+\mathbb{E}\left[ K_{\tilde{G}}\right]} \geq  \frac{N+ 3\mathbb{E}\left[ K_{\tilde{G}}\right]}{2+\mathbb{E}\left[ K_{\tilde{G}}\right]}
\end{equation*}
where we used that $\mathbb{E}\left[ K_{H}^2\right]\geq \mathbb{E}\left[ K_{H}\right]$ for any $H\in \mathcal{G}^{ud}$ since $K_H$ is non-negative integer valued. Note that the function $f(x)=\frac{N+3x}{2+x}$, $x>-2$, non-decreasing for $N\leq 6$ and decreasing for $N\geq 7$. Also note that $\mathbb{E}\left[ K_{\tilde{G}}\right]$ ranges between $0$ and $N-1$. Hence, for $N\leq 6$ we obtain that $Q(G)\geq N/2$ and for $N\geq 7$ we deduce that $Q(G)\geq (4N-1)/(N+1)$. 

\end{enumerate}
\end{proof}

Note that $(4N-1)/(N+1)$ is increasing in $N$ and larger than $3$ for $N\geq 7$, so there are sequences $(l_N)_{N\in\mathbb{N}}$ simultaneously satisfying the constraints given in 3.\ and 4.\ of Proposition~\ref{prop:maxtotdegreeundir}, and, of course, also 2.

\subsubsection{DMFNR Based on Degree Centrality}
Let $\mathfrak{C}^{deg}(v, G) = k^{G}_v$ for $v\in \mathcal{V}_G$ (and $\mathfrak{C}^{deg}(v, G) = \emptyset$ if $v\notin \mathcal{V}_G$). 
\begin{proposition}\label{prop:maxtotdegreeundir}
Consider a set $\mathcal{A}\subset\mathcal{G}^{ud}$ as in \eqref{eq:accsetundir} with $Q(G) = \max_{v\in\mathcal{V}_G} \mathfrak{C}^{deg} (v, G)$. Then 
\begin{enumerate}
\item $Q$ is $\mathcal{I}$-monotone for any $\mathcal{I}\subset \mathcal{I}^{ud}_{e\_del}\cup \mathcal{I}^{ud}_{split}$. Hence $(\mathcal{A}, \mathcal{I}, \mathcal{C})$ is a DMFNR for any cost function for $(\mathcal{A}, \mathcal{I})$. 
\item $\mathcal{A}$ satisfies P1 if and only if $l_N\geq 0$ for all $N\geq 2$,
 \item $\mathcal{A}$ satisfies  P2 if and only if there is a $N_0$ such that $l_N\geq 2$ for all $N\geq N_0$.
 \item $\mathcal{A}$ satisfies  P3 if and only if $l_N < N-1$ for all $N\geq 1$.
\end{enumerate}
\end{proposition}
\begin{proof}
  1. and 2. are straightforward modifications of the corresponding results in Propositions \ref{prop:degmon} and \ref{prop:outindegree}. \\
3.: Note that any undirected tree graph of size $N\geq 3$ contains at least one node with a total degree of at least 2, and as before the tree graphs represent the connected graphs with minimal $Q$ in $\mathcal{G}^{ud}$ for a given network size $N$ since $Q(G)$ is $\mathcal{I}_{e\_del}$-monotone. This shows the 'only if'-part, and for the 'if'-part consider undirected line graphs. \\
    4.:  For a star node $v^\ast$ in an undirected star graph $G^\ast$ we obtain $$k_{v^\ast}^{G^\ast} = \frac{1}{2} \big(k_{v^\ast}^{G^\ast, in} + k_{v^\ast}^{G^\ast, out}\big) = \frac{1}{2} (\vert\mathcal{V}_G\vert-1 + \vert\mathcal{V}_G\vert-1) = \vert\mathcal{V}_G\vert-1.$$ 
\end{proof}

\subsubsection{DMFNR Based on Closeness Centrality}
Consider $Q(G) = \max_{v\in\mathcal{V}_G} \mathfrak{C}^{close}(v, G)$ where  $$\mathfrak{C}^{close}(v, G) := \frac{1}{2} \big(\mathfrak{C}^{close}_{in}(v, G) + \mathfrak{C}^{close}_{out}(v, G)\big).$$

\begin{proposition}\label{prop:closenessundir}
Consider a set $\mathcal{A}\subset\mathcal{G}^{ud}$ as in \eqref{eq:accsetundir} with $Q(G) = \max_{v\in\mathcal{V}_G}\mathfrak{C}^{close}(v, G)$. Then
\begin{enumerate}
\item $Q$ is $\mathcal{I}$-monotone for any $\mathcal{I}\subset \mathcal{I}^{ud}_{e\_del}\cup \mathcal{I}^{ud}_{s\_iso}$. Hence $(\mathcal{A}, \mathcal{I}, \mathcal{C})$ is a DMFNR for any cost function for $(\mathcal{A}, \mathcal{I})$. 
\item $\mathcal{A}$ satisfies P1 if and only if $l_N\geq 0$ for all $N\geq 2$,
 \item $\mathcal{A}$ satisfies P2 if and only if there is a $N_0$ such that $l_N\geq (1/(N-1)) \sum_{j=1}^{(N-1)/2} (2/j)$ for all $N\geq N_0$ which are odd, and $l_N\geq (1/(N-1)) \big(\sum_{j=1}^{(N-2)/2} (2/j) + 2/N\big)$ for all $N\geq N_0$ which are even.
\item $\mathcal{A}$ satisfies  P3 if and only if $l_1<0$ and $l_N < 1$ for all $N\geq 2$.
\end{enumerate}
\end{proposition}
\begin{proof}
    See the proof of Proposition \ref{prop:closeness} for 1. and 2. For 4., note that the star node $v^\ast$ in an undirected star graph $G^\ast$ of size $|\mathcal{V}_{G^\ast}|\geq 2$ comes with $\mathfrak{C}^{close}_{in}(v^\ast, G^\ast) = \mathfrak{C}^{close}_{out}(v^\ast, G^\ast) = 1$, and thus $Q(G^\ast) =1$.
    
   Regarding 3., it suffices again to restrict the discussion to tree graphs since $Q$ is $\mathcal{I}_{e\_ del}$-monotone. Consider a tree graph $G^T$ of size $N$, and let $v,w\in\mathcal{V}_{G^T}$ be two nodes (not necessarily unique) that come with a maximal distance $l^{G^T}_{vw} (= l^{G^T}_{wv})\leq N-1$ among all node pairs. Moreover, note that the path $p^{G^T}_{vw} := (v, v_2,\cdots , w)$ connecting $v$ and $w$ in $G^T$ is unique since a tree graph contains no cycles. Now we distinguish between two cases: 
    \begin{enumerate}
        \item Suppose $l^{G^T}_{vw}$ is even. Then the number of nodes on the path $p^{G^T}_{vw}$ is odd. Thus we find a node $u\in p^{G^T}_{vw}$ that lies in the center of the path connecting $v$ and $w$, coming with $l^{G^T}_{vu} (=l^{G^T}_{uv}) = l^{G^T}_{wu}  (=l^{G^T}_{uw}) = l^{G^T}_{vw}/2$. 
        Moreover, since $v$ and $w$ come with the maximal distance among all node pairs, we have $l^{G^T}_{uq}\leq l^{G^T}_{vw}/2$ for all nodes $q\in\mathcal{V}_{G^T}$:
        \begin{enumerate}
            \item The statement is clear for each node $q$ that lies on the path connecting $v$ and $w$. 
            \item Suppose the path $p^{G^T}_{vq}$ between $q$ and $v$ or $p^{G^T}_{wq}$ between $q$ and $w$ does not intersect with the path $p^{G^T}_{vw}$ connecting $v$ and $w$ (apart from the node $v$ or $w$). W.l.o.g. assume $p^{G^T}_{vq}\setminus\{v\}\cap p^{G^T}_{vw} = \emptyset$. Since all node pairs in $G^T$ are connected by exactly one path due to the absence of cycles in $G^T$, we thus have $l^{G^T}_{wq} = l^{G^T}_{vw} + l^{G^T}_{vq} > l^{G^T}_{vw}$. But this contradicts the assumption that  $v$ and $w$ come with the maximal distance among all node pairs. 
            \item For the last case, consider the possibility that the paths connecting node $q$ to $v$ and to $w$ depart from a node $r\in p^{G^T}_{vw}$. Then $l^{G^T}_{vq} = l^{G^T}_{vr} + l^{G^T}_{rq}$ and $l^{G^T}_{wq} = l^{G^T}_{wr} + l^{G^T}_{rq}$. We can now either have $r=u$, or if not, then for either $v$ or $w$, namely the node with a larger distance to $q$,  $u$ must lie on the path connecting $v$ and $q$, or $w$ and $q$, respectively. If we had the case that $l^{G^T}_{uq}> l^{G^T}_{vw}/2$, then again, since $G^T$ contains no cycles, the path between $v$ and $q$ or between $w$ and $q$ now must be larger than $l^{G^T}_{vw}$, contradicting our initial assumption.
        \end{enumerate}
        Now, since $u$ lies in the centre of the path between $v$ and $w$, for any distance $j\leq l^{G^T}_{vw}/2$, we find two different nodes $x,y$ on the path $p^{G^T}_{vw}$ with $l^{G^T}_{ux}, l^{G^T}_{uy} = j$. Moreover, due to  $l^{G^T}_{uq}\leq l^{G^T}_{vw}/2$ for all $q\in\mathcal{V}_{G^T}$ we obtain 
        \begin{equation*}
            \sum_{q\in\mathcal{V}_{G^T}\setminus\{ u\}}\frac{1}{l^{G^T}_{uq}}\geq \sum_{j=1}^{l^{G^T}_{vw}/2} 2\cdot\frac{1}{j} + \sum_{q\notin p^{G^T}_{vw}} \frac{1}{l^{G^T}_{vw}/2} = \sum_{j=1}^{l^{G^T}_{vw}/2} 2\cdot\frac{1}{j} + (N-l^{G^T}_{vw}-1)\cdot \frac{1}{l^{G^T}_{vw}/2}
        \end{equation*}
        If $N$ is odd, then $N-l^{G^T}_{vw}-1$ is even, and because of $l^{G^T}_{vw}\leq N-1$ we then we find 
         \begin{equation*}
            \sum_{q\in\mathcal{V}_{G^T}\setminus\{ u\}}\frac{1}{l^{G^T}_{uq}}\geq \sum_{j=1}^{(N-1)/2} 2\cdot\frac{1}{j}.
        \end{equation*}
        
        In the case that $N$ is even, we find that $N-l^{G^T}_{vw}-1$ is odd, thus
         \begin{equation*}
            \sum_{q\in\mathcal{V}_{G^T}\setminus\{ u\}}\frac{1}{l^{G^T}_{uq}}\geq \sum_{j=1}^{(N-2)/2} 2\cdot\frac{1}{j} + \frac{1}{l^{G^T}_{vw}/2} \geq \sum_{j=1}^{(N-2)/2} 2\cdot\frac{1}{j} + \frac{2}{N-1} \geq \sum_{j=1}^{(N-2)/2} 2\cdot\frac{1}{j} + \frac{2}{N}.
        \end{equation*}
    \item Suppose that $l^{G^T}_{vw}$ is odd. Then we find two nodes $u_v$ and $u_w$ in the center of the path between $v$ and $w$, where $l^{G^T}_{v u_v}, l^{G^T}_{w u_w} = (l_{vw}-1)/2$, and $l^{G^T}_{w u_v}, l^{G^T}_{v u_w} = (l_{vw}+1)/2$. Moreover, analogously to the considerations from 1., we then find that $l^{G^T}_{u_v q}, l^{G^T}_{u_w q} \leq (l^{G^T}_{vw}+1)/2 \leq N/2$ for all nodes $q\in\mathcal{V}_{G^T}$. Thus, for $u\in\{ u_v, u_w\}$ we have
       \begin{equation*}
            \sum_{q\in\mathcal{V}_{G^t}\setminus\{ u\}}\frac{1}{l^{G^T}_{uq}} = \sum_{j=1}^{(l^{G^T}_{vw}-1)/2} 2\cdot\frac{1}{j} + \frac{1}{(l^{G^T}_{vw}+1)/2} + \sum_{q\notin p^{G^T}_{vw}} \frac{1}{l^{G^T}_{uq}} \geq \sum_{j=1}^{(l^{G^T}_{vw}-1)/2} \frac{2}{j} + (N-l^{G^T}_{vw})\cdot \frac{1}{(l^{G^T}_{vw}+1)/2}.
        \end{equation*}  
        Now, if $N$ is odd, then we must have $l^{G^T}_{u_v q}, l^{G^T}_{u_w q} \leq (l^{G^T}_{vw}+1)/2 \leq (N-1)/2$, and    $N-l^{G^T}_{vw}$ is even, which together yields
         \begin{equation*}
            \sum_{q\in\mathcal{V}_{G^T}\setminus\{ u\}}\frac{1}{l^{G^T}_{uq}}\geq \sum_{j=1}^{(N-1)/2} \frac{2}{j}.
        \end{equation*}
        If $N$ is even, then $N-l^{G^T}_{vw}$ is odd, and therefore, we can estimate
        
          \begin{equation*}
            \sum_{q\in\mathcal{V}_{G^T}\setminus\{ u\}}\frac{1}{l^{G^T}_{uq}}\geq \sum_{j=1}^{(N-1)/2} \frac{2}{j} + \frac{1}{N/2} = \sum_{j=1}^{(N-1)/2} \frac{2}{j} + \frac{2}{N}.
        \end{equation*}
    \end{enumerate}
    Finally, for the line graph of size $N$ we indeed have  
     \begin{equation*}
            \sum_{q\in\mathcal{V}_{G^L}\setminus\{ u\}}\frac{1}{l^{G^L}_{uq}} = \sum_{j=1}^{(N-1)/2} \frac{2}{j}
        \end{equation*}
        if $N$ is odd (which then is equivalent to $l^{G^L}_{vw}$ being even). In the case that $N$ is even (and thus $l^{G^L}_{vw}$ is odd), then 
        \begin{equation*}
            \sum_{q\in\mathcal{V}_{G^L}\setminus\{ u\}}\frac{1}{l^{G^L}_{uq}}= \sum_{j=1}^{(N-1)/2} \frac{2}{j} + \frac{2}{N}.
        \end{equation*}

\end{proof}

\subsubsection{Control of the Spectral Radius of Undirected Graphs}\label{sec:specrad} 

Dynamic systems on networks are usually described by operators that depend on the adjacency matrix $A_G$ of the graph $G$. The linear analysis of steady states and stability properties of these states typically leads to an analysis of the spectral properties of the adjacency matrix. Of particular importance here is the largest eigenvalue $\lambda_1^G$ of $A_G$, called the \textit{spectral radius} of $G$. For details, we refer to Section 17.4 in \cite{Newman2018}. 

\begin{example}\label{ex:SIS}
    A commonly applied model for the propagation of contagious threats in (cyber) networks is the SIS (\textnormal{Susceptible-Infected-Susceptible}) Markov process, see Section 3.2 in \cite{Awiszus2023a}. The single node infection dynamics of the system can be described by a set of ordinary differential equations, according to Kolmogorov's forward equation. A major problem as regards the solvability of this system of ordinary differential equations is the fact that this system is not closed as higher order moments appear, cascading up to the network size. A simplified model can be obtained when assuming that all infection probabilities in the network are pairwise uncorrelated. Given a  network $G$ of size $N$ with enumerated nodes, this  so-called \textnormal{NIMFA approximation} or \textnormal{individual-based model} is governed by the set of equations
    \begin{equation}\label{eq:NIMFA}
\frac{d \mathbb{E}[I_i(t)]}{d t} = \tau\Big( \sum_{j=1}^N A_G(i,j) \big(1- \mathbb{E}[I_i(t)]\big) \mathbb{E}[I_j(t)]\Big) - \gamma \mathbb{E}[I_i(t)],\quad i=1,\ldots , N,
\end{equation}
where $\tau, \gamma >0$ are the infection and recovery rates, and $I_i(t) = 1$ if node $v_i$ is infected at time  $t$, and $I_i(t) = 0$  if node $v_i$ is susceptible at time $t$.
The infection probabilities of this model provide upper bounds for the actual infection probabilities and can therefore be used for a conservative estimation, see Theorem 3.3 in \cite{Kiss2017}.
The stability of steady states of system \eqref{eq:NIMFA}, in particular of the disease-free state $(I_1,\cdots , I_N) = (0,\cdots, 0)$, is now relevant in the analysis of the epidemic vulnerability of the network $G$. Linearizing the individual-based model from \eqref{eq:NIMFA} around the disease-free state induces the eigenvalue problem 
\begin{equation*}
    \det (\tau A_G -\gamma \mathbb{I}_N - \lambda \mathbb{I}_N) = 0
\end{equation*}
where $\mathbb{I}_N$ is the $N$-dimensional identity matrix and the largest eigenvalue of $\tau A_G - \gamma \mathbb{I}_N$ is given by $\tau \lambda_1^G-\gamma$. If $G$ is undirected and connected, it can be shown that if $\tau/\gamma < 1/\lambda_1^G$, then the disease free state is stable and no endemic state exists; if $\tau /\gamma > 1/\lambda_1^G$, then the disease-free state is unstable and there exists a unique endemic state, which is stable. For details, we refer to Section 3.4.4 in \cite{Kiss2017} and Section 17.4 in \cite{Mieghem2014}.
\end{example}

\begin{remark}
 In the setting of directed graphs, however, the information of the graph spectrum may have little or no relevance in terms of graph vulnerability. For example, the only element in the spectrum of both the edgeless and the directed star graph is =. Consequently, Properties P1 and P3 cannot jointly be satisfied in a directed graph setting when choosing $Q(G) =\lambda_1^G$ for acceptance set $\mathcal{A}$. Moreover, the spectrum of directed graphs can be complex due to the lack of the adjacency matrix' symmetry.
 \end{remark}
 The following definition is needed for the proof of the next proposition:
\begin{definition}\label{def:kelman}
Given a graph $G\in\mathcal{G}^{ud}$ and two nodes $v, u\in\mathcal{V}_G$, a \textit{Kelmans operation} in $G$ is a $\mathcal{I}_{shift}$-strategy where each undirected edge $\{(v, w), (w,v)\}\in\mathcal{E}_G$ is replaced by $\{(u,w), (w,u)\}$ if $w\neq u$ and $\{(u,w), (w,u)\}\cap\mathcal{E}_G=\emptyset$. 
\end{definition} 

\begin{proposition}\label{prop:spectr}
    Consider a set $\mathcal{A}\subset\mathcal{G}^{ud}$ as in \eqref{eq:accsetundir} with  $Q(G) = {\lambda}_1^G$. Then 
\begin{enumerate}
\item $Q$ is $\mathcal{I}$-monotone for any $\mathcal{I}\subset \mathcal{I}^{ud}_{e\_del}\cup \mathcal{I}^{ud}_{split}$. Hence $(\mathcal{A}, \mathcal{I}, \mathcal{C})$ is a DMFNR for any cost function for $(\mathcal{A}, \mathcal{I})$. 
\item $\mathcal{A}$ satisfies P1 if and only if $l_N\geq 0$ for all $N\geq 2$,
 \item $\mathcal{A}$ satisfies  P2 if and only if there is a $N_0$ such that $l_N\geq \max_{j\in \{ 1, \cdots, N\}} \cos(\pi j / (N+1))$ for all $N\geq N_0$. In particular, this is the case if there is a $N_0$ such that $l_N\geq 2$ for all $N\geq N_0$.
 \item $\mathcal{A}$ satisfies  P3 if and only if $l_N < \sqrt{(N-1)}$ for all $N\geq 1$.
\end{enumerate}
\end{proposition}

\begin{proof}
    \begin{enumerate}
        \item $Q$ is $\mathcal{I}^{ud}_{e\_ del}$-monotone according to \cite[Proposition 3.1.1]{Brouwer2012}. Standard results show that $Q$ is also $\mathcal{I}^{ud}_{split}$-monotone: Let $G\in\mathcal{G}^{ud}$, and $H=\kappa_{split}^{\mathcal{J}, v, \tilde{v}}(G)$ a network resulting from a split of node $v$ by adding a new node $\tilde{v}$ and rewiring edges contained in $\mathcal{J}$, where $(v,w)\in\mathcal{J}\Leftrightarrow (w,v)\in\mathcal{J}$. Then we can construct a new graph $H'$ from $H$ by a consecutive edge shift where all edges $(\tilde{v}, w)$ are replaced by $(v,w)$ and all $(w, \tilde{v})$ by $(w, v)$. This is a particular example of a  Kelmans operation, see Section 3.1.3 in \cite{Brouwer2012} for details, since $\mathcal{N}^H_v\cap\mathcal{N}^H_{\tilde{v}}=\emptyset$. Note that $H' = G\dot\cup (\{\tilde{v}\}, \emptyset)$ is a disjoint union of $G$ and the isolated node $\tilde{v}$. Therefore, we have  ${\lambda}_1^{H'} = {\lambda}_1^G$, and the spectral radius of a graph is non-decreasing under a Kelmans operation, see \cite[Proposition 3.1.5]{Brouwer2012}. This gives ${\lambda}_1^H\leq{\lambda}_1^{H'} = {\lambda}_1^G$.
    \item trivial
    \item Note that undirected line graphs of size $N$ have a minimal spectral radius among all connected graphs of this size in $\mathcal{G}^{ud}$, cf. \cite[Lemma 1]{Dam2007}. Their spectrum  consists of the eigenvalues $\lambda_j = 2\cos (\pi j / (N+1)) $, $j=1,\cdots, N$, see \cite[Section 1.4.4]{Brouwer2012}. 
    \item For an undirected star graph of size $N$, simple calculations show that the spectrum consists of the eigenvalues $\sqrt{N-1}$ and $-\sqrt{N-1}$, both with a multiplicity of one, and the eigenvalue 0 with multiplicity of $N-2$.
    \end{enumerate}
\end{proof}

\bibliographystyle{comnet}
\bibliography{ReferencesCyberProject}
%








\end{document}